\documentclass[oribibl]{llncs}

\usepackage{enumitem}
\usepackage{amsmath, amssymb, myMath}
\usepackage{graphicx}
\usepackage[caption=false]{subfig}
\usepackage{stackrel}
\usepackage{comment}

\title{MSOL-Definability Equals Recognizability for Halin Graphs and Bounded
Degree $k$-Outerplanar Graphs\protect\footnote{The
research of the second author was partially funded by the Networks programme,
funded by the Dutch Ministry of Education, Culture and Science through the Netherlands
Organisation for Scientific Research.}}

\author{Lars Jaffke\protect\footnote{Department of Information and Computing
Sciences, Utrecht University, P.O. Box 80.089, 3508 TB Utrecht, The Netherlands. Email:
\email{l.jaffke@students.uu.nl}} 
\and 
Hans L. Bodlaender\protect\footnote{Department of
Information and Computing Sciences, Utrecht University, P.O. Box 80.089, 3508 TB
Utrecht, The Netherlands.
Department of Mathematics and Computer Science,
University of Technology Eindhoven,
P.O. Box 513, 5600 MB Eindhoven, The Netherlands.
Email: \email{h.l.bodlaender@uu.nl}}}

\institute{}

\begin{document}

\maketitle

\begin{abstract}
	One of the most famous algorithmic meta-theorems states that every graph
	property that can be defined by a sentence in counting monadic second order
	logic (CMSOL) can be checked in linear time for graphs of bounded treewidth,
	which is known as Courcelle's Theorem \cite{Cou90}.
	These algorithms are constructed as finite state tree automata, and hence every
	CMSOL-definable graph property is recognizable. Courcelle also conjectured that
	the converse holds, i.e.\ every recognizable graph property is definable in
	CMSOL for graphs of bounded treewidth. We prove this conjecture
	for a number of special cases in a stronger form. That is, we show that
	each recognizable property is definable in MSOL, i.e.\ the counting operation
	is not needed in our expressions. We give proofs for Halin graphs, bounded
	degree $k$-outerplanar graphs and some related graph classes.
	We furthermore show that the conjecture holds for any graph class that admits
	tree decompositions that can be defined in MSOL, thus providing a useful tool
	for future proofs.
\end{abstract}

\section{Introduction}\label{secIntro}
In a seminal paper from 1976, Rudolf Halin (1934-2014), lay the ground work
for the notion of tree decompositions of graphs \cite{Hal76}, which later
was studied deeply in the proof of the famous Graph Minor Theorem by Robertson
and Seymour \cite{RS84} and ever since became one of the most important tools
for the design of FPT-algorithms for NP-hard problems on graphs.
He was also the first one to extensively study the class of planar graphs
constructed by a tree and adding a cycle through all its leaves, now known as Halin graphs
\cite{Hal71}. \par
Another seminal result is Courcelle's Theorem \cite{Cou90}, which states that
for every graph property $P$ that can be formulated in a language called
counting monadic second order logic (CMSOL), and each fixed $k$, there is a
linear time algorithm that decides $P$ for a graph given a tree decomposition of
width at most $k$ (while similar results were discovered by Arnborg et al.
\cite{ALS91} and Borie et al. \cite{BPT92}).
Counting monadic second order logic generalizes monadic second order logic
(MSOL) with a collection of predicates testing the size of sets modulo
constants. Courcelle showed that this makes the logic strictly more powerful
\cite{Cou90}, which can be seen in the following example.
\begin{example}
	Let $P$ denote the property that a graph has an even number of vertices. Then
	$P$ is trivially definable in CMSOL, but it is not in MSOL.
\end{example}
The algorithms constructed in Courcelle's proof have the shape of a finite state
tree automaton and hence we can say that CMSOL-definable graph properties are
recognizable (or, equivalently, regular or finite-state).
Courcelle's Theorem generalizes one direction of a classic result in automata
theory by B{\"u}chi, which states that a language is recognizable,
if and only if it is MSOL-definable \cite{Bue60}.
Courcelle conjectured in 1990 that the other direction of B{\"u}chi's result can
also be generalized for graphs of bounded treewidth in CMSOL, i.e.\ that each
recognizable graph property is CMSOL-definable. \par
This conjecture is still regarded to be open. Its claimed resolution by
Lapoire \cite{Lap98} is not considered to be valid by several experts.
In the course of time proofs were given for the classes of trees and forests 
\cite{Cou90}, partial 2-trees \cite{Cou91}, partial 3-trees and $k$-connected
partial $k$-trees \cite{Kal00}. A sketch of a proof for graphs of pathwidth
at most $k$ appeared at ICALP 1997 \cite{Kab97}. Very recently, one of the
authors proved, in collaboration with Heggernes and Telle, that Courcelle's
Conjecture holds for partial $k$-trees without chordless cycles of length at
least $\ell$ \cite{BHT15}. 
\par 
In this paper we give self-contained proofs for Halin graphs, $k$-outerplanar
graphs of bounded degree, a subclass of $k$-outerplanar graphs (of unbounded
degree) and some classes related to feedback edge and/or vertex sets of bounded
size w.r.t.\ a spanning tree in the graph.
In all of these cases we show a somewhat stronger result, as we restrict ourselves to
MSOL-definability, thus avoiding the above mentioned counting predicate. Since
Halin graphs have treewidth 3 \cite{Wim87}, Kaller's result implies
that recognizable properties are CMSOL-definable in this case \cite{Kal00}. We
strengthen this result to MSOL-definability. \par
Additionally, we show that Courcelle's Conjecture holds in our stronger sense
for each graph class that admits certain types of MSOL-definable tree
decompositions.
We believe that this technique provides a useful tool towards its resolution ---
if not for all graph classes, then at least for a significant number of special
cases.
\par
In our proofs, we use another classic result in automata theory, the
Myhill-Nerode Theory \cite{Myh57}\cite{Ner58}. It states that a language $L$ is
recognizable if and only if there exists an equivalence relation $\sim_L$,
describing $L$, that has a finite number of equivalence classes (i.e. $\sim_L$
has \emph{finite index}). 
Abrahamson and Fellows \cite{AF93} noted that the Myhill-Nerode Theorem can also
be generalized to graphs of bounded treewidth (see also \cite[Theorem
12.7.2]{DF13}): Each graph property $P$ is recognizable if and only if there
exists an equivalence relation $\sim_P$ of finite index, describing $P$,
defined over terminal graphs with a bounded number of terminal vertices.
This result was recently generalized to hypergraphs \cite{BFGR13}. \par
The general outline of our proofs can be described as follows. Given a graph
property $P$, we assume the existence of an equivalence relation $\sim_P$ of
finite index. We then show that, given a tree decomposition of bounded width, we
can derive the equivalence classes of terminal subgraphs w.r.t.\ its nodes from
the equivalence classes of their children. Once we reach the root of the tree
decomposition we can decide whether a graph has property $P$ by the equivalence
class its terminal subgraph is contained in. We then show that this construction
is MSOL-definable.
\par
The rest of the paper is organized as follows. In Section \ref{secPrel}, we give
the basic definitions and explain all concepts that we use in more
detail. In Section \ref{secConstEQ} we prove some technical results
regarding equivalence classes w.r.t.\ nodes in tree decompositions. The main
results are presented in Sections \ref{secHalin} and \ref{secGen}, where we
prove Courcelle's Conjecture for Halin graphs and other graph classes, such as bounded
degree $k$-outerplanar graphs. We give some concluding remarks in Section
\ref{secConc}.

\section{Preliminaries}\label{secPrel}
\subsection*{Graphs and Tree Decompositions}
We begin by giving the basic definitions of the graph classes and some related
concepts used throughout the paper. 
\begin{definition}[(Planar) Embedding]
	A drawing of a graph in the plane is called an \emph{embedding}. If no pair of
	edges in this drawing crosses, then it is called \emph{planar}.
\end{definition}
\begin{definition}[Halin Graph]
	A graph is called a \emph{Halin graph}, if it can be formed by a planar
	embedding of a tree, none of whose vertices has degree two, and a cycle that
	connects all leaves of the tree such that the embedding stays planar.
\end{definition}
\begin{definition}[$k$-Outerplanar Graph]
	Let $G = (V, E)$ be a graph. $G$ is called a \emph{planar
	graph}, if there exists a planar embedding of $G$. \par
	An embedding of a graph $G$ is \emph{$1$-outerplanar}, if it is
	planar, and all vertices lie on the exterior face. For $k \ge 2$, an embedding
	of a graph $G$ is \emph{$k$-outerplanar}, if it is planar, and when
	all vertices on the outer face are deleted, then one obtains a
	$(k-1)$-outerplanar embedding of the resulting graph. If $G$ admits a
	$k$-outerplanar embedding, then it is called a \emph{$k$-outerplanar graph}.
\end{definition}
One can immediately establish a connection between the two graph
classes.
\begin{proposition}
	Halin graphs are 2-outerplanar graphs.
\end{proposition}
The following definition will play a central role in many of the proofs of
Sections \ref{secHalin} and \ref{secGen}.
\begin{definition}[Fundamental Cycle]\label{defFundCyc}
	Let $G = (V, E)$ be a graph with maximal spanning forest $T = (V, F)$. Given an
	edge $e = \{v, w\}$, $e \in E \setminus F$, its \emph{fundamental cycle} is a
	cycle that is formed by the unique path from $v$ to $w$ in $F$ together with
	the edge $e$.
\end{definition}
We now turn to the notion of tree decompositions and some related concepts.
\begin{definition}[Tree Decomposition, Treewidth]
	A \emph{tree decomposition} of a graph $G = (V, E)$ is a pair $(T, X)$ of a
	tree $T = (N, F)$ and an indexed family of vertex sets $(X_t)_{t \in N}$
	(called \emph{bags}), such that the following properties hold.
	\begin{enumerate}[label=(\roman*)]
	  \item Each vertex $v \in V$ is contained in at least one bag.
	  \item For each edge $e \in E$ there exists a bag containing both endpoints.
	  \item For each vertex $v \in V$, the bags in the tree decomposition that
	  contain $v$ form a subtree of $T$.
	\end{enumerate}
	The \emph{width} of a tree decomposition is the size of the largest bag minus 1
	and the \emph{treewidth} of a graph is the minimum width of all its tree
	decompositions.
\end{definition}
To avoid confusion, in the following we will refer to elements of $N$ as
\emph{nodes} and elements of $V$ as \emph{vertices}. Sometimes, to shorten the
notation, we might not differ between the terms \emph{node} and \emph{bag} in a
tree decomposition.
\begin{definition}[Node Types]\label{defBagTypes}
	We distinguish three types of nodes in a tree decomposition $(T,
	X)$, listed below.
	\begin{enumerate}[label={(\roman*)}]
	  \item The nodes corresponding to leaves in $T$ are called \emph{Leaf} nodes.
	  \item If a node has exactly one child it is called an \emph{Intermediate}
	  node.
	  \item If a node has more than one child it is called a \emph{Branch} node.
	\end{enumerate}
\end{definition}
As we will typically speak of some direction between nodes in tree
decompositions, such as a parent-child relation, we define the following.
\begin{definition}[Rooted and Ordered Tree Decomposition]
	Consider a tree decomposition $(T = (N, F), X)$. We call $(T, X)$
	\emph{rooted}, if there is one distinguished node $r \in N$, called the
	\emph{root} of $T$, inducing a parent-child relation on all edges in $F$.
	If there exists a fixed ordering on all bags sharing the same parent node,
	then $T$ is called \emph{ordered}.
\end{definition}
We now introduce \emph{terminal graphs}, over which we will later define
equivalence relations for graph properties.
\begin{definition}[Terminal Graph]
	A \emph{terminal graph} $G = (V, E, X)$ is a graph with vertex set $V$, edge
	set $E$ and an \emph{ordered} terminal set $X \subseteq V$.
\end{definition}
Terminal graphs of special interest in the rest of this paper are \emph{terminal
subgraphs} w.r.t.\ bags in a tree decomposition.
We require the notion of \emph{partial} terminal subgraphs in the proofs of
Sections \ref{secConstEQ} and \ref{secTDMSOL}.
\begin{definition}[(Partial) Terminal Subgraph]\label{defTermSG}
	Let $(T = (N, F), X)$ be a rooted (and ordered) tree decomposition of a graph
	$G = (V, E)$ with bags $X_t$ and $Y_{t'}$, $t, t' \in N$, such that $t$ is the
	parent node of $t'$. The graphs defined below are induced subgraphs of $G$
	given the respective vertex sets.
	\begin{enumerate}[label={(\roman*)}]
	  \item A \emph{terminal subgraph} of a bag $X_t$, denoted by $\termSG{X_t}$,
	  is a terminal graph induced by the vertices in $X_t$ and all its descendants,
	  with the set $X_t$ as its terminals.\label{defTermSGTS}
	  \item A \emph{partial terminal subgraph} of $X_t$ \emph{given a child}
	  $Y_{t'}$, denoted by $\pTermSG{X_t}{Y_{t'}}$ is the terminal graph induced by
	  $X_t$ and the vertices and edges of all terminal subgraphs of the children of
	  $X_t$ that are left siblings of $Y_{t'}$, with terminal set
	  $X_t$.\label{defTermSGPTS}
	\end{enumerate}
	The ordering in each terminal set of the above mentioned terminal graphs can be
	arbitrary, but fixed.
\end{definition}
For an illustration of Definition \ref{defTermSG}, see Figure \ref{figEQCJoin1},
where $H = \termSG{X_H}$ and $G = \pTermSG{X_G}{X_H}$.

\subsection*{Equivalence Relations}
\begin{definition}[Gluing via $\oplus$]\label{defGlue}
	Let $G = (V_G, E_G, X_G)$ and $H = (V_H, E_H, X_H)$ be two terminal graphs with
	$|X_G| = |X_H|$. The graph $G \oplus H$ is obtained by taking the disjoint
	union of $G$ and $H$
	and for each $i$, $1 \le i \le |X_G|$, identifying the $i$-th vertex in $X_G$
	with the $i$-th vertex in $X_H$.
\end{definition}
Note that if an edge is included both in $G$ and in $H$, we drop one
of the edges in $G \oplus H$, i.e. we do not have parallel edges in the graph.
\par
We use the operator $\oplus$ to define equivalence relations over
terminal graphs. Throughout the paper we will restrict ourselves to terminal
graphs of fixed boundary size (i.e. the maximum size of terminal sets is
bounded by some constant), since we focus on equivalence relations with a finite
number of equivalence classes. These, in general, do not exist for classes of
terminal graphs with arbitrary boundary size (see \cite{AF93}).
\begin{definition}[Equivalence Relation over Terminal Graphs]\label{defEQC}
	Let $P$ denote a graph property. $\sim_P$ denotes an equivalence relation
	over terminal graphs, \emph{describing} $P$, defined as follows. Let $G$, $H$
	and $K$ be terminal graphs with fixed boundary size. Then we have:
	\begin{equation*}
		G \sim_P H \Leftrightarrow \forall K : P(G \oplus K) \Leftrightarrow P(H
		\oplus K)
	\end{equation*}
	This yields notions of \emph{equivalence classes} and
	\emph{finite index} in the ordinary way.  
	We might drop the index $P$ in case it is clear from the context.
\end{definition}
We illustrate Definition \ref{defEQC} with an example.
\begin{example}
	Let $P$ denote the property that a graph has a Hamiltonian cycle. Let $G$ and
	$H$ be two terminal graphs with terminal sets $X_G$ and $X_H$, respectively
	(where $|X_G| = |X_H|$).
	We say that $G$ and $H$ are equivalent w.r.t. $\sim_P$, if for all terminal
	graphs $K$ (with terminal set $X_K$, $|X_K| = |X_G| = |X_H|$), the graph $G
	\oplus K$ contains a Hamiltonian cycle if and only if $H \oplus K$ contains a
	Hamiltonian cycle. A simple case when this hols is when both
	$G$ and $H$ contain a Hamiltonian path such that their terminal sets consist of
	the two endpoints of the path.
\end{example}
As mentioned earlier, our ideas are based on the Myhill-Nerode Theory for graphs
of bounded treewidth. The following theorem formally states this result.
\begin{theorem}[Myhill-Nerode Theorem for Graphs of Treewidth
$k$]\label{thmMyhNerTWK} 
Let $P$ denote a graph property. Then the following are equivalent for any fixed
$k$.
	\begin{enumerate}[label={(\roman*)}]
		\item $P$ is recognizable for graphs of treewidth at most
		$k$.\label{thmMyhNerTWK1}
		\item There exists an equivalence relation $\sim_P$, describing $P$, of finite
		index.\label{thmMyhNerTWK3}
	\end{enumerate}
\end{theorem}
By the proof of this theorem (see, e.g., \cite[Theorem 12.7.2]{DF13}) we know
that we can identify some equivalence classes of $\sim_P$ with accepting states
in the automaton given in \ref{thmMyhNerTWK1}. Let $C_P$
denote such an ('accepting') equivalence class and $G \in C_P$ a terminal graph.
Then we know that the graph $G \oplus (X_G, \emptyset, X_G)$ has property $P$. We will use
this fact in the proofs of Sections \ref{secFIIDHalin} and \ref{secTDMSOL}.

\subsection*{MSOL-Definability}
We now define monadic second order logic over graphs. All variables that we use
in our expressions are either single vertices/edges or vertex/edge sets.
\emph{Atomic predicates} are logical statements with the least number of
variables, e.g. the vertex membership '$v \in V$'.
Higher-order predicates can be formed by joining predicates
via negation $\neg$, conjunction $\wedge$, disjunction $\vee$, implication
$\rightarrow$ and equivalence $\leftrightarrow$ together with the existential
quantifier $\exists$ and the universal quantifier $\forall$. A predicate without
free variables, i.e. variables that are not in the scope of some quantifier, is
called a \emph{sentence}.
A graph property is called \emph{MSOL-definable} if we can express it with an
MSOL-sentence.
\par
A central concept used in this paper is an implicit representation
of a tree decomposition in monadic second order logic, as we cannot refer to
bags and edges in a tree decomposition as variables in MSOL directly.
Hence, we most importantly require
two types of predicates. The first one will
allow us to verify whether a vertex is contained in some bag and whether any
vertex set in the graph constitutes a bag in its tree decomposition.
In our definition, each bag will be associated with either a vertex or an edge
in the underlying graph together with some \emph{type}, whose definition depends
on the actual graph class under consideration.
The second one allows for identifying edges in the tree decomposition, i.e. for
any two vertex sets $X$ and $Y$, this predicate will be true if and only if both $X$ and
$Y$ are bags in the tree decomposition and $X$ is the bag corresponding to the
parent node of $Y$. \par
While all MSOL-definable tree decompositions have to be rooted, not all of them
have to be ordered. In some cases, however, an ordering on nodes with the same
parent is another prerequisite, which also has to be verifiable with an
MSOL-predicate.
\begin{definition}[MSOL-definable tree decomposition]
	A rooted (and ordered) tree decomposition $(T, X)$ of a graph $G$ is called
	\emph{MSOL-definable}, if the following hold.
	\begin{enumerate}[label={(\roman*)}]
	  \item Each bag $X$ in the tree decomposition can be identified by one of the
	  following predicates (where $s$ and $t$ are constants). 
	  \begin{enumerate}[label={(\alph*)}]
	    \item $\Bag_{\tau_1}^V(v, X),\ldots,\Bag_{\tau_t}^V(v, X)$: The bag $X$ is
	    associated with type $\tau_i$ and the vertex $v \in V$, where $1 \le i \le
	    t$.
	    \item $\Bag_{\sigma_1}^E(e, X),\ldots,\Bag_{\sigma_s}^E(e, X)$: The bag $X$
	    is associated with type $\sigma_j$ and the edge $e \in E$, where $1 \le j
	    \le s$.
	  \end{enumerate}
	  Furthermore there
	  exists at least one type that contains the corresponding vertex or both
	  endpoints of the corresponding edge.
		\item There exists a predicate $\Parent(X_p, X_c)$ to identify edges in $T$,
		which is true, if and only if $X_p$ is the parent bag of $X_c$.
	\end{enumerate}
	We call an MSOL-definable tree decomposition \emph{ordered}, if the following
	holds. 
	\begin{enumerate}[label={(\roman*)}]
	  \setcounter{enumi}{2}
	  \item There exists a predicate $\oriNBA(X_l, X_r)$, which is true if and only
	  if $X_l$ and $X_r$ are siblings such that $X_l$ is the direct left sibling of
		$X_r$.
	\end{enumerate}
\end{definition}

\section{Constructing Equivalence Classes}\label{secConstEQ}
The current section contains a number of technical results related to
equivalence classes of (partial) terminal subgraphs of bags in a tree
decomposition.
In particular, we will show how to derive the equivalence classes of
(partial) terminal subgraphs of bags in a tree decomposition from the
equivalence classes of some (partial) terminal subgraphs of child/sibling bags.
Hence we prove that these equivalence classes are related to each other in the
same way as states in some finite automaton via its transition function, which
will be of vital importance in the proofs of Sections \ref{secFIIDHalin} and
\ref{secTDMSOL}.
\par
In the following, unless stated otherwise, we assume that our tree decomposition
is rooted and ordered. First, we consider branch nodes. We begin by defining an
operator, which can be seen as an extension of the $\oplus$-operator.
\begin{definition}[Gluing via $\oplus_\rhd$]
	Let $X_G$ be a branch bag in a tree decomposition with child $X_H$ and let $G
	= \pTermSG{X_G}{X_H} = (V_G, E_G, X_G)$ and $H = \termSG{X_H} = (V_H, E_H,
	X_H)$ denote the partial terminal subgraph of $X_G$ given $X_H$ and the
	terminal subgraph of $X_H$, respectively. The operation $\oplus_\rhd$ is
	defined as:
	\begin{equation*}
		G \oplus_\rhd H = (V_G \cup V_H, E_G \cup E_H, X_G)
	\end{equation*}
\end{definition}
Note that again, we drop parallel edges, if they occur.
\par
\begin{figure}
	\centering
	\subfloat[The respective
	terminal graphs]{\includegraphics[width=.45\textwidth]{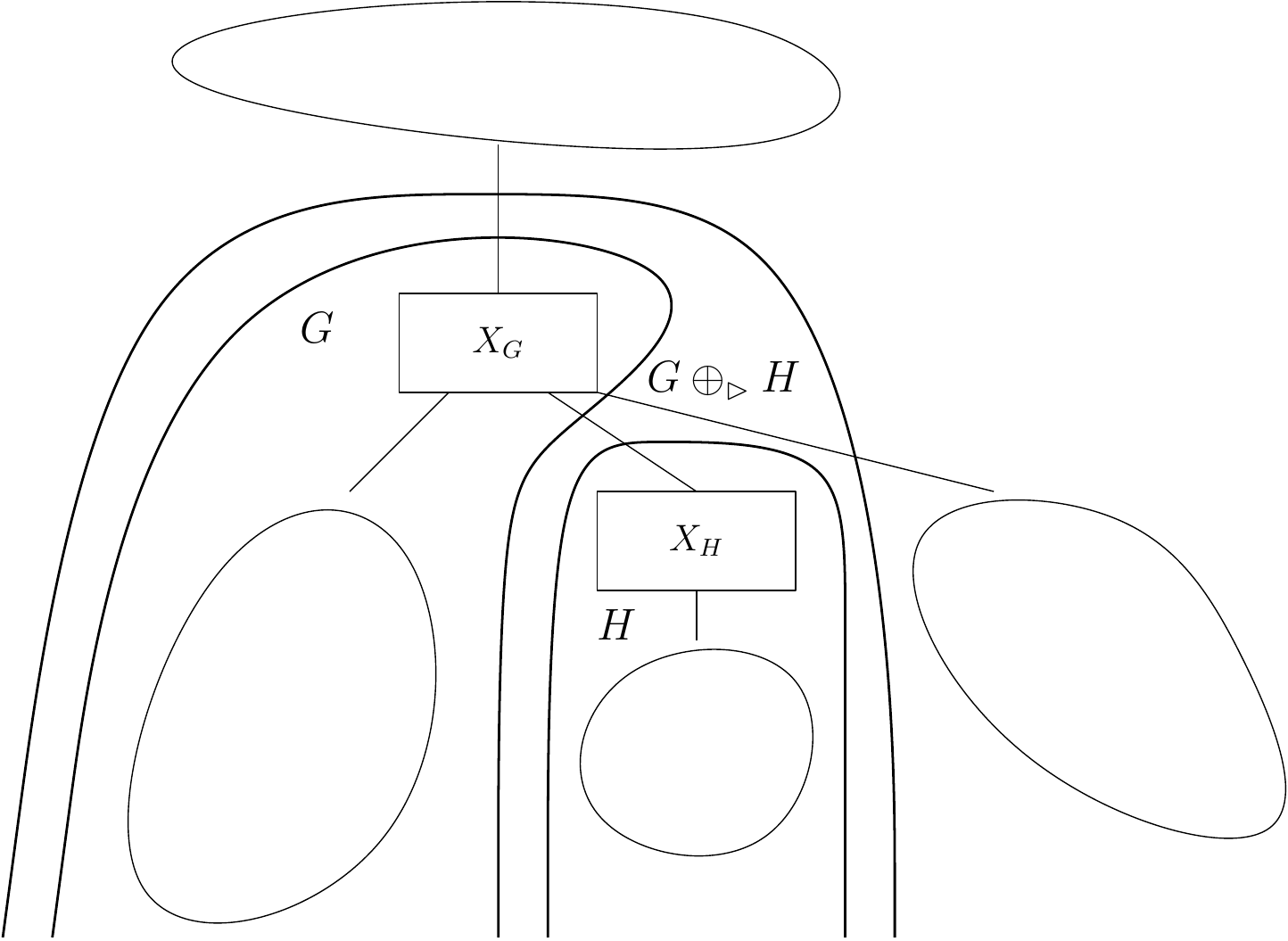}
	\label{figEQCJoin1}}
	\qquad
	\subfloat[Splitting $\{X_{G}, X_{H}\}$]{
	\includegraphics[width=.45\textwidth]{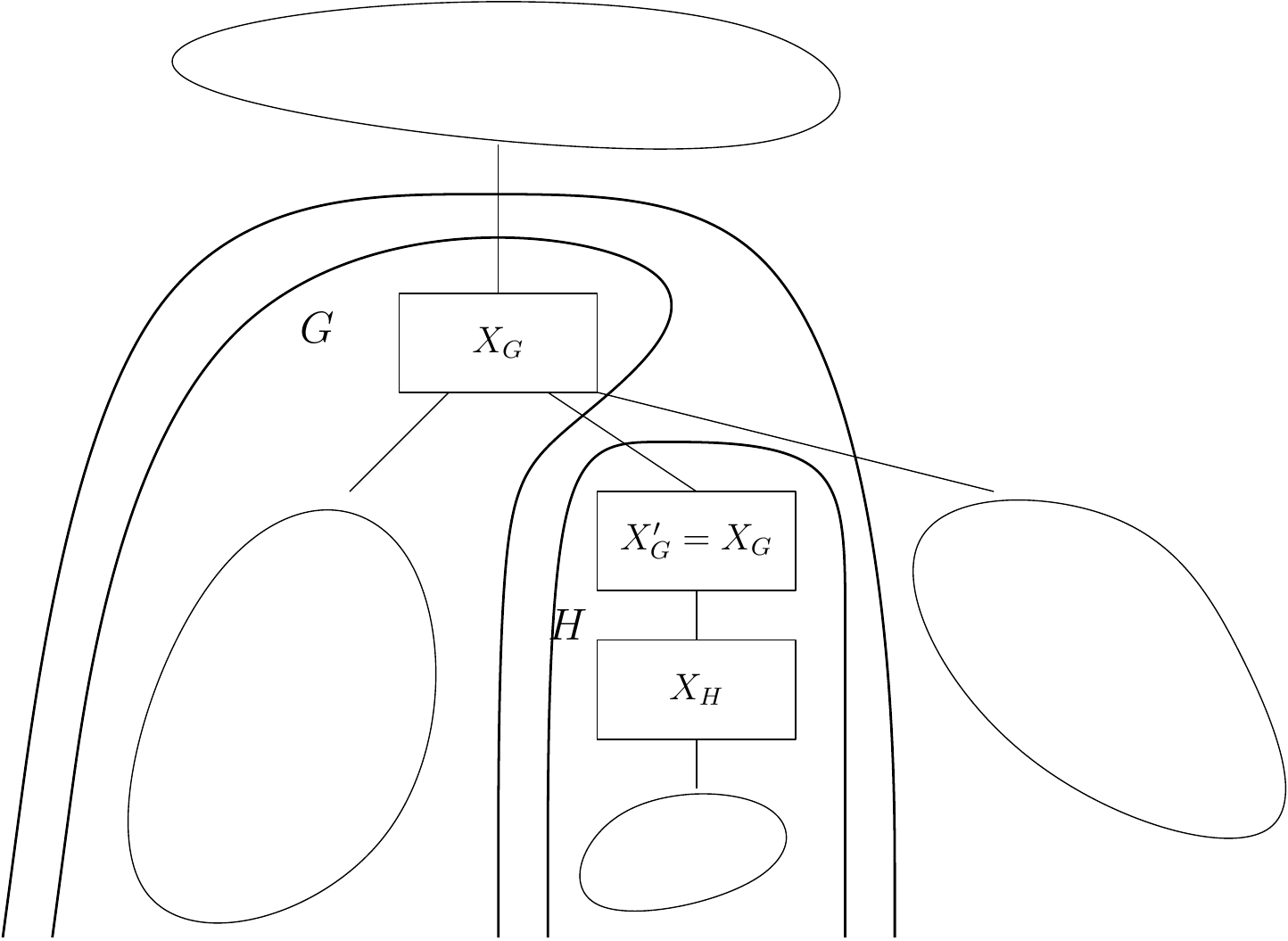} 
	\label{figEQCJoin2}}
	\caption{Branch node in a tree decomposition}
	\label{figEQCJoin}
\end{figure}
Consider the situation depicted in Figure \ref{figEQCJoin} and suppose that we
know the equivalence class for the graph $G = \pTermSG{X_G}{X_H}$, i.e. the
partial terminal subgraph of $X_G$ given $X_H$, and the equivalence class for
graph $H = \termSG{X_H}$, the terminal subgraph of $X_H$. We want to
derive the equivalence class of the partial terminal subgraph of $X_G$ given
the right sibling of $X_H$ (which is the terminal graph $G \oplus_\rhd H$).
\par
We will prove that the equivalence class of $G \oplus_\rhd H$ only depends on
the equivalence class of $G$ and $H$ by explaining how we can create a terminal
graph in this class from any pair of graphs $G' \sim G$, $H' \sim H$ with
$X_{G'} = X_G$ and $X_{H'} = X_H$.
Note that since we are only interested in determining whether the underlying
graph of the tree decomposition, say $G^*$, has property $P$, it is sufficient
to only consider terminal graphs in the equivalence classes of $G$ and $H$ that have the
same terminal sets as $G$ and $H$. These classes contain any number of
(terminal) graphs, which are (also up to isomorphism) completely unrelated to
$G^*$ and hence can be disregarded.
The following lemma formalizes the above discussion.
\begin{lemma}\label{lemEquivDerJoin}
	Let $X_G$ be a branch bag in a tree decomposition and $X_H$ one of its child
	bags. Let $G = \pTermSG{X_G}{X_H}$, $H = \termSG{X_H}$ and $G'$ and $H'$ two
	terminal graphs.
	If $G' \sim G$, $H' \sim H$, $X_G =X_{G'}$ and $X_H = X_{H'}$, then $(G
	\oplus_\rhd H) \sim (G' \oplus_\rhd H')$.
\end{lemma}
\begin{proof}
We first define an operator that allows us to rewrite $\oplus_\rhd$.
\begin{definition}[Gluing via $\oplus_T$]
	Let $G$ be a (terminal) graph and $X$ an ordered set of vertices. The operation
	$\oplus_T$ is defined as:
	\begin{equation*}
		G \oplus_T X = (V_G \cup X, E_G, X)
	\end{equation*}
	That is, we take the (not necessarily disjoint) union of $X$ and the vertices
	in $G$ and let $X$ be the terminal set of the resulting terminal graph.
\end{definition}
Note that $\oplus_T$ can either be used to make a graph a terminal graph, or
to equip a terminal graph with a new terminal set. 
One easily observes the following.
\begin{proposition}\label{propOplusRhdRewrite}
	Let $G$ and $H$ be two terminal graphs as in Lemma \ref{lemEquivDerJoin}.
	Then,
	\begin{equation}
		G \oplus_\rhd H = \overbrace{(G \oplus \underbrace{(H \oplus_T X_G)}_{(b)})
		\oplus_T X_G}^{(a)}\label{eqOplusRhdRewrite}.
	\end{equation}
\end{proposition}
This process of rewriting $\oplus_\rhd$ can be illustrated as shown in Figure
\ref{figEQCJoin2}. Instead of computing $G \oplus_\rhd H$ directly, we split the
edge between the bags $X_G$ and $X_H$, creating a new bag $X_G'$ in between the
edge, where $X_G' = X_G$. Then we extend $H$ to a terminal graph with terminal
set $X_G'$ by using the $\oplus_T$-operator. Denote this graph by $H_{X_G'}$.
Since $H_{X_G'}$ has terminal set $X_G' = X_G$, we can apply $\oplus$ to $G$ and
$H'$, such that all vertices that are identified in the operation are equal. This results
in the graph consisting of all vertices and edges in both $G$ and $H$.
Eventually, we apply $\oplus_T$ to the resulting graph again to make it a
terminal graph with terminal set $X_G$. \par
We will lead the proof of Lemma \ref{lemEquivDerJoin} in two steps: First we
show that we can construct graphs equivalent to $(G \oplus H) \oplus_T X_G$ by
members of the equivalence classes of $G$ and $H$, if $G$ and $H$ have the same
terminal set (Part (a) of Equation \ref{eqOplusRhdRewrite}, where $H$ denotes the
terminal graph $H \oplus_T X_G$).
In the second step, we show that we can construct graphs equivalent to $H
\oplus_T X$ from members of the equivalence class $H$ for any terminal set $X$
(Part (b) of Equation \ref{eqOplusRhdRewrite}).
\par We now proceed with the formal proofs.
\begin{figure}[t]
	\centering
	\includegraphics[width=.5\textwidth]{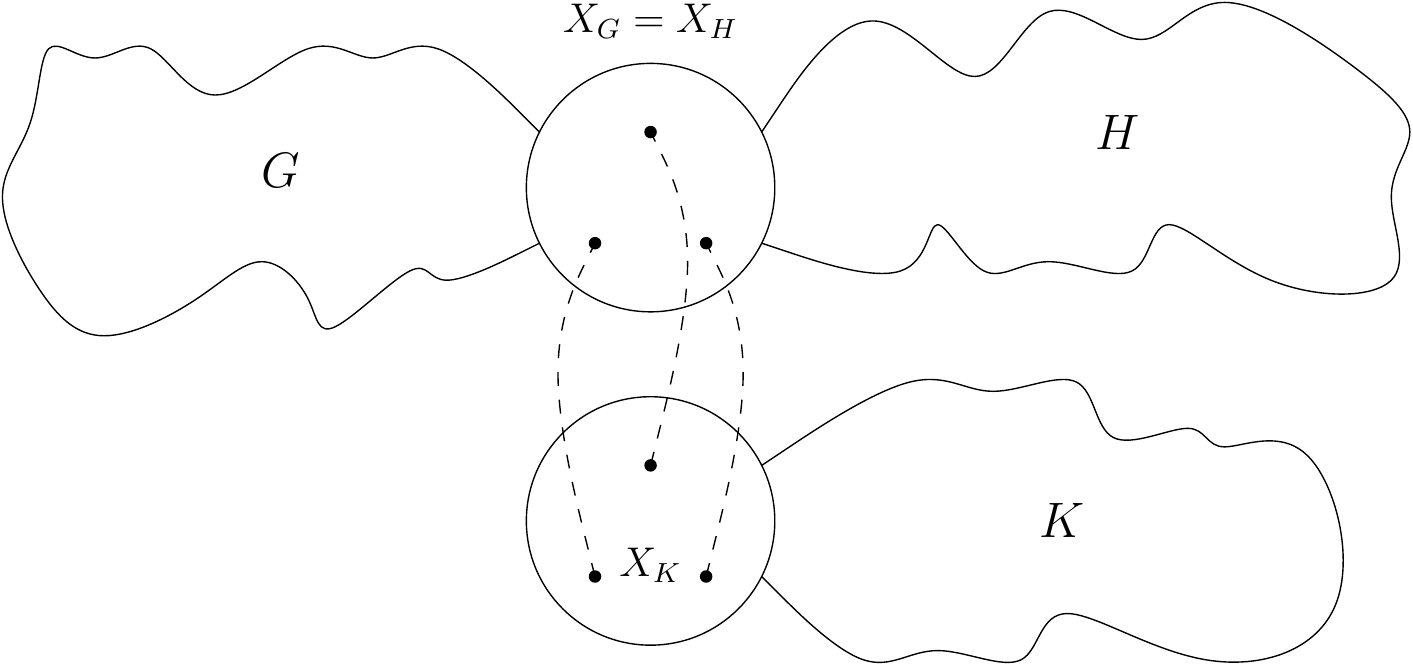}
	\caption{Terminal graphs $G, H$ and $K$ as in the proof of Proposition
	\ref{propOplusGlue}. The dashed lines indicate, which vertices are being
	identified in the corresponding $\oplus$-operation.}
	\label{figOplusProof}
\end{figure}
\begin{proposition}\label{propOplusGlue}
	Let $G = (V_G, E_G, X_G)$ and $H = (V_H, E_H, X_H)$ be two terminal graphs with
	$X_G = X_H$. Let $G'$ and $H'$ be two terminal graphs with $G' \sim G$, $H'
	\sim H$, $X_G = X_{G'}$ and $X_H = X_{H'}$.
	Then,
	\begin{equation*}
		(G \oplus H) \oplus_T X_G \sim (G' \oplus H') \oplus_T X_{G'}.
	\end{equation*} 
\end{proposition}
\begin{proof}
	By Figure \ref{figOplusProof} we can observe the following.
 	\begin{equation*}
		K \oplus ((G \oplus H) \oplus_T X_G) = G \oplus ((K \oplus H) \oplus_T X_G)
 	\end{equation*}
	Regardless of the order in which we apply the operators, both graphs will have
	the same vertex and edge sets.
	As for the identifying step (using the $\oplus$-operator), one can see that for
	all $i = 1,\ldots,|X_K|$ we have that the $i$-th vertex in $X_K$ is identified
	with the $i$-th vertex in $X_G$ in the left-hand side of the equation and with
	the $i$-th vertex in $X_H$ in the right-hand side. The equality still holds,
	since $X_G = X_H$. We use this argument (and the fact that $X_{G'} = X_G =
	X_{H} = X_{H'}$) to show the following.
	\begin{align*}
		\forall K : ~&P(K \oplus ((G \oplus H) \oplus_T X_G)) \Leftrightarrow P(G
		\oplus ((K \oplus H) \oplus_T X_G)) \\
		\Leftrightarrow &P(G' \oplus ((K \oplus H) \oplus_T X_{G'})) \Leftrightarrow
		P (H \oplus ((K \oplus G') \oplus_T X_{H})) \\
		\Leftrightarrow &P(H' \oplus ((K \oplus G') \oplus_T X_{H'})) \Leftrightarrow 
		P(K \oplus ((G' \oplus H') \oplus_T X_{G'})) \\
	\end{align*}
	Hence, our claim follows.
\qed \end{proof}
\begin{figure}[t]
	\centering
	\includegraphics[width=.7\textwidth]{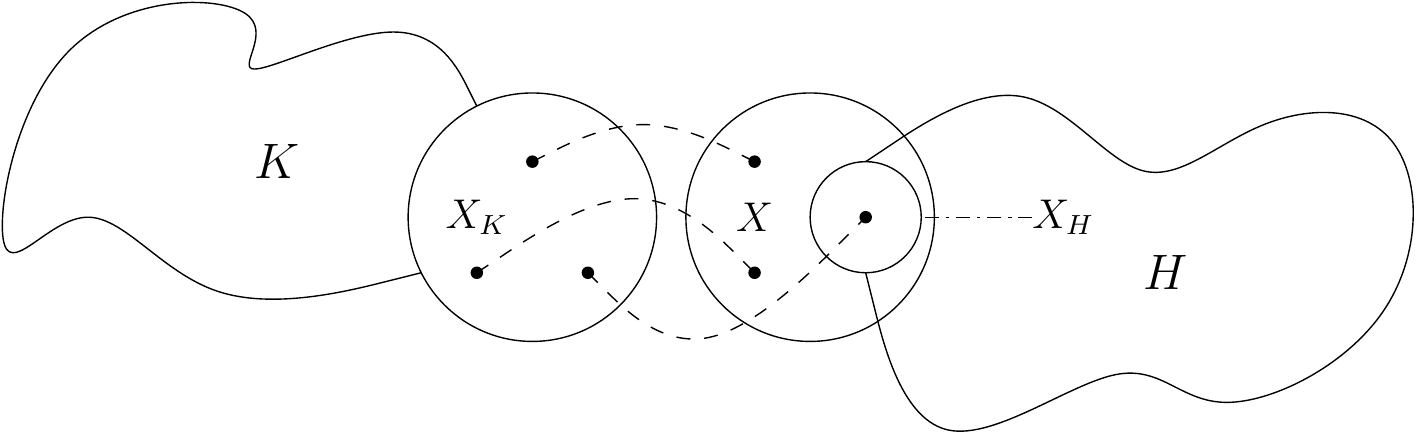}
	\caption{Terminal graphs $H$ and $K$, and a terminal set $X$. The dashed
	lines indicate, which vertices are being identified in the corresponding
	$\oplus$-operation.}
	\label{figOplusTProof}
\end{figure}
\begin{lemma}\label{lemOplusTGlue}
	Let $H, H'$ be terminal graphs with $H \sim H'$, $X_H = X_{H'}$ and $X$ an
	ordered vertex set.
	Then, $H \oplus_T X \sim H' \oplus_T X$.
\end{lemma}
\begin{proof}
	By Figure \ref{figOplusTProof}, one can derive a similar argument as in the
	proof of Proposition \ref{propOplusGlue}. Note that $|X_K| = |X|$ (otherwise,
	$\oplus$ is not defined) and let $K_X = K \oplus (X, \emptyset, X)$, i.e. the
	graph obtained by identifying each $i$-th vertex in $X_K$ with each $i$-th
	vertex in $X$, where $1 \le i \le |X_K|$. Then,
	\begin{equation*}
	  K \oplus (H \oplus_T X) = H \oplus (K_X \oplus_T X_H).
	\end{equation*}
	In the left-hand side, we first extend the terminal graph $H$ to have terminal
	set $X$ and then glue the resulting graph to $K$. Thus the $i$-th vertex in
	$X_K$ is identified with the $i$-th vertex in $X$, $i = 1,\ldots,|X_K|$. The
	same vertices are being identified in the first step in computing the
	right-hand side, which is constructing the graph $K_X$. We then extend this
	graph to have terminal set $X_H$ and glue it to the graph $H$. Since again, in both
	of the computations the same vertices get identified and both graphs have equal
	vertex and edge sets, we see that our claim holds.
	We use this argument (and the fact that $X_H = X_{H'}$) to conclude our proof
	as follows.
	\begin{align*}
		\forall K :~& P(K \oplus (H \oplus_T X)) \Leftrightarrow P(H \oplus (K_X
		\oplus_T X_H)) \\
		\Leftrightarrow &P(H' \oplus (K_X \oplus_T X_{H'})) \Leftrightarrow P(K \oplus
		(H' \oplus_T X))
	\end{align*}
\qed \end{proof}
This concludes our proof of Lemma \ref{lemEquivDerJoin}.
\qed \end{proof}
The methods used in this proof also allow us
to handle intermediate nodes in a tree decomposition. For an illustration see
Figure \ref{figEQCInt}.
Lemma \ref{lemOplusTGlue} suffices as an argument that we can derive the
equivalence class of $G$ from graphs equivalent to $H$.
\par
Next, we generalize the situation of Lemma \ref{lemEquivDerJoin},
where we were dealing with two child nodes of a branch bag, to handle any
constant number of children at a time (see Figure \ref{figEQCJoinBDeg}). We will
apply this result to tree decompositions that are not ordered but instead have
bounded degree.
\begin{figure}[t]
	\centering
	\subfloat[Intermediate node, where $G = H \oplus_T X_G$.]{
			\includegraphics[width=.255\textwidth]{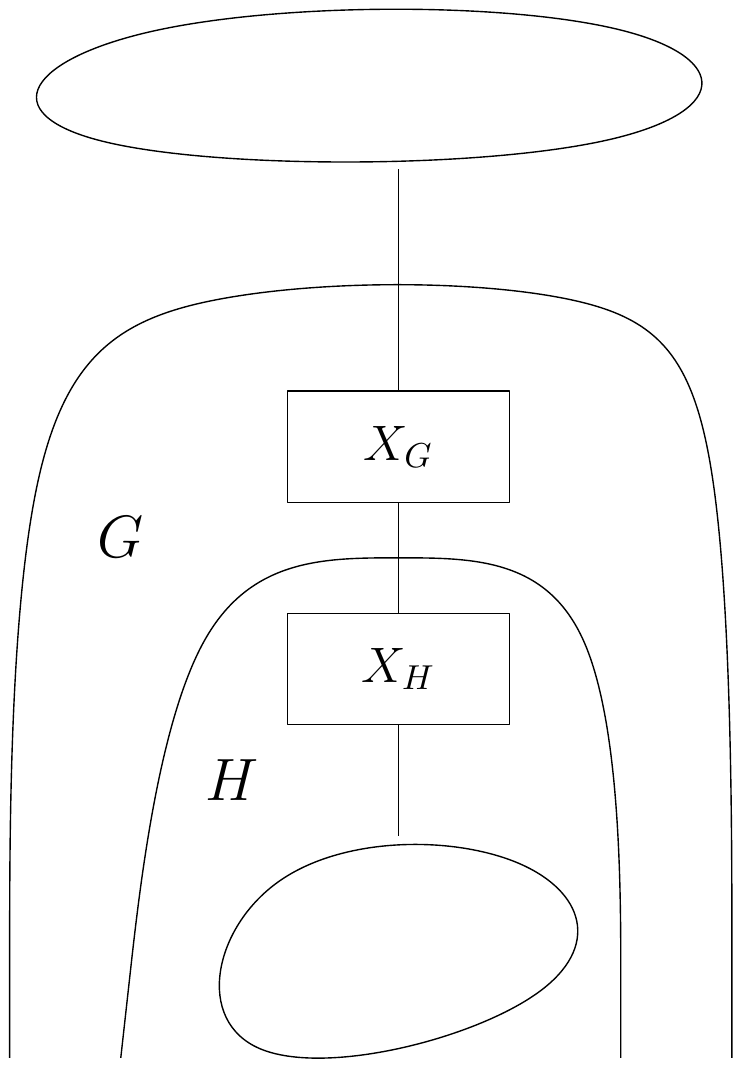}
			\label{figEQCInt}
		}
	\qquad
	\subfloat[Bounded degree branch node. Note that $G = (H_1 \oplus_T X_G)
	\oplus_\rhd (H_2 \oplus_T X_G) \oplus_\rhd (H_3 \oplus_T X_G)$.]{
		\includegraphics[width=.555\textwidth]{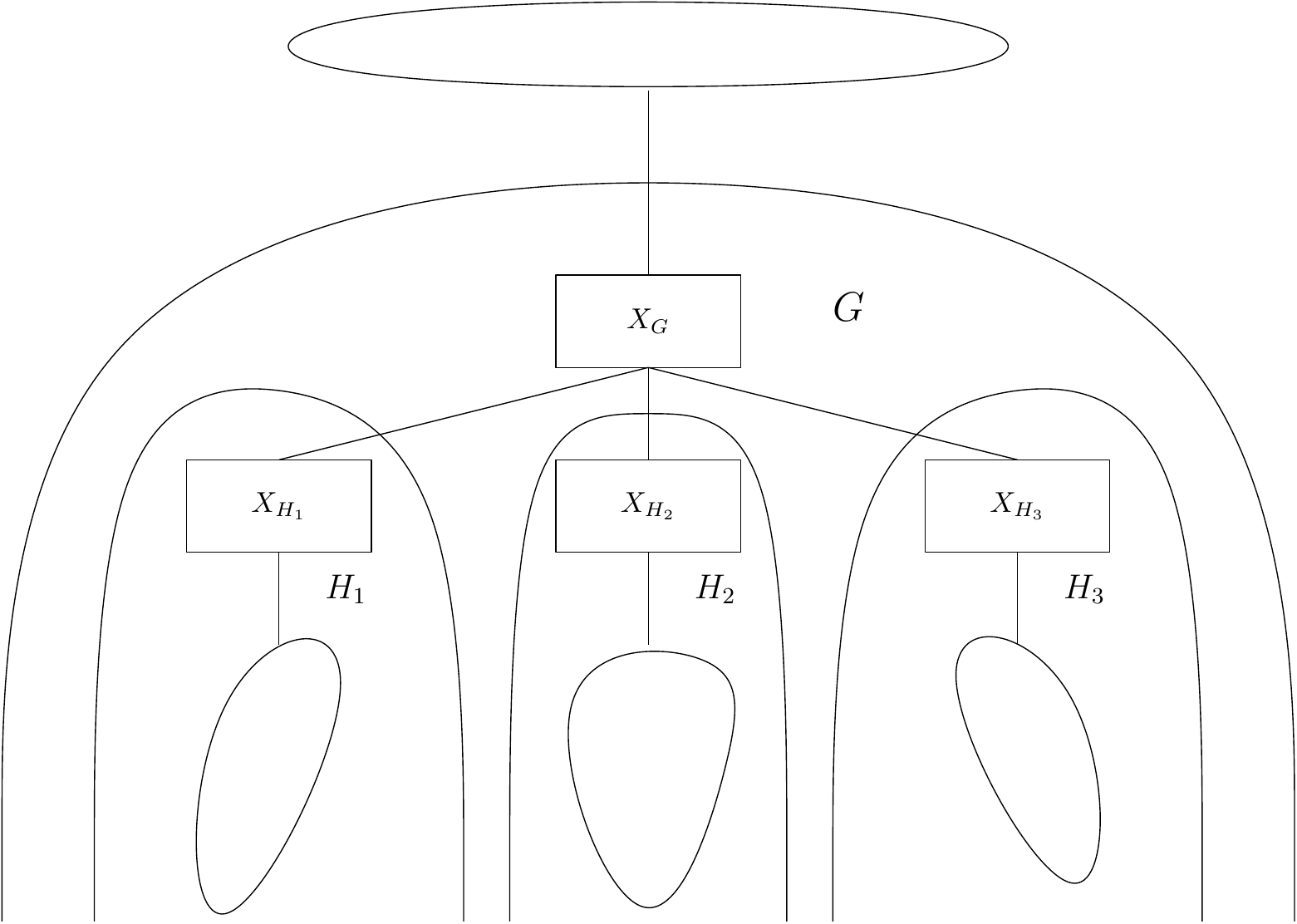}
		\label{figEQCJoinBDeg}
	}
	\caption{Intermediate and bounded degree branch node in a tree decomposition.}
\end{figure}
\begin{lemma}\label{lemEQCJoinConstCh}
	Let $X_G$ be a branch bag in a tree decomposition with a constant number of
	child bags $X_{H_1},\ldots,X_{H_c}$. Let $H_1 = \termSG{X_{H_1}}$,\ldots, $H_c
	= \termSG{X_c}$. If $H_1' \sim H_1,\ldots,H_c' \sim H_c$ and $X_{H_1'} =
	X_{H_1},\ldots,X_{H_c'} = X_{H_c}$, then
	\begin{align*}
		(H_1 \oplus_T X_G) \oplus_\rhd \cdots \oplus_\rhd (H_c \oplus_T X_G)
		\sim (H_1' \oplus_T X_G) \oplus_\rhd \cdots \oplus_\rhd (H_c' \oplus_T X_G)
	\end{align*}
\end{lemma}
\begin{proof}
	Let $G$ and $H$ be the two terminal graphs as indicated below.
	\begin{equation*}
		\underbrace{(H_1 \oplus_T X_G)}_{G} \oplus_\rhd \underbrace{(H_2 \oplus_T
		X_G) \oplus_\rhd \cdots \oplus_\rhd (H_c \oplus_T X_G)}_{H}
	\end{equation*}
	Since $H_1 \sim H_1'$, we know by Lemma \ref{lemOplusTGlue}, that $(H_1
	\oplus_T X_G) \sim (H_1' \oplus_T X_G)$. Let $G' = (H_1' \oplus_T X_G)$, then
	we have that $G \sim G'$. Now, by Lemma \ref{lemEquivDerJoin}, we know that
	$(G \oplus_\rhd H) \sim (G' \oplus_\rhd H)$ and hence:
	\begin{equation*}
		G \oplus_\rhd H \sim (H_1' \oplus_T X_G) \oplus_\rhd H
	\end{equation*}
	We can apply this argument repeatedly and our claim follows. Note that the
	child bags $X_{H_1},\ldots,X_{H_c}$ do not need a specific ordering, as in this
	context the operation $\oplus_\rhd$ is commutative (all graphs, which it is
	applied to, have terminal set $X_G$).
\qed \end{proof}

\section{Halin Graphs}\label{secHalin}
This section is devoted to proving our first main result, which is that
MSOL-definability equals recognizability for the class of Halin graphs. As
outlined before, we will prove that finite index implies
MSOL-definability. In a first step, we will show that we can define a certain
orientation on the edges of a Halin graph together with an ordering on edges
with the same head vertex in monadic second order logic (Section
\ref{secOrdHG}), which we then will use to construct MSOL-definable tree
decompositions of Halin graphs (Section \ref{secTDHG}). We conclude the
proof in Section \ref{secFIIDHalin}. \par
In many of the proofs of MSOL-definability of graph (or tree
decomposition) properties, we use other graph properties that have been shown to
be MSOL-definable before, and refer for more precise expressions to
the appendix.

\subsection{Edge Orientation and Ordering}\label{secOrdHG}
In the following we will develop an orientation on the edges of a
Halin graph, together with an ordering on edges with the head vertex, which is
MSOL-definable. Our goal is that in this orientation, the edges that form the
cycle connecting the leaves is a directed cycle and the tree of the
Halin graph forms a directed tree with some arbitrary root on the outer cycle.
\begin{lemma}[Cf. \cite{Cou95}, Lemma 4.8 in \cite{Kal00}]\label{lemTWKOrd}
	Let $G$ be a graph of treewidth $k$. Any orientation $\msolOri$ on its edges
	using predicates $\head(e, v)$ and $\tail(e, v)$ is MSOL-definable.
\end{lemma}
\begin{proof}
	Since $G$ has treewidth $k$, we know that it admits a $k+1$-coloring on its
	vertices.
	We assume we are given such a coloring and denote the color set by
	$\{0,1,\ldots,k\}$. Now let $F$ be a set of edges of $G$ and $e = \{v, w\}$
	an edge in the graph. We know that $\col(v) \neq \col(w)$ and thus we either
	have $\col(v) < \col(w)$ or $\col(v) > \col(w)$. We let the edge $e$ be
	directed from $v$ to $w$, if
	\begin{enumerate}[label={(\roman*)}]
	  \item $\col(v) < \col(w)$ and $e \in F$, or
	  \item $\col(v) > \col(w)$ and $e \notin F$
	\end{enumerate}
	and otherwise from $w$ to $v$. Thus we can choose any orientation of the edge
	set of $G$ by choosing the corresponding set $F$. Assuming that $\msolOri$
	uses predicates $\head(e, v)$ and $\tail(e, v)$ as shown in Appendix
	\ref{appSecMSOLHalinOrd}, we can define our sentence as
	\begin{equation*}
		\exists X_0 \cdots \exists X_k (\exists F \subseteq E)(\kCol{k+1}(V,
		X_0,\ldots,X_k) \wedge \msolOri).
	\end{equation*}
\qed \end{proof}
\begin{lemma}\label{lemHalinEdgOri}
	Let $G = (V, E)$ be a Halin graph. The orientation on the edge set of $G$ such
	that its spanning tree forms a rooted directed tree and the outer cycle is a directed
	cycle, is MSOL-definable.
\end{lemma}
\begin{proof}
	Since Halin graphs have treewidth 3, we can use Lemma \ref{lemTWKOrd}.
	Let $E_T$ denote the edges in the spanning tree and $E_C$ the edges on the
	outer cycle. The orientation stated above can be defined in MSOL as
	\begin{equation*}
		\msolOri = \exists E_T \exists E_C (\Part_E(E, E_T, E_C) \wedge \dir{\Tree}(V,
		E_T) \wedge \dir{\Cycle}(\IncV(E_C), E_C)).
	\end{equation*}
	The MSOL-predicates given in Appendix \ref{appSecMSOLHalinOrd} complete the
	proof.
\qed \end{proof}
Next, we define an ordering on all edges with the same head vertex in a Halin
graph, which we can define in monadic second order logic using the orientation
of the edges given above and its fundamental cycles. This is a central step in
our proof, as it allows us to avoid using the counting predicate in the
construction of our tree decomposition. The main idea in the proof of Lemma
\ref{lemHalinNBOrd} is that we can order the child edges of a vertex in the
order in which their leaf descendants appear on the outer cycle.
\begin{lemma}\label{lemHalinNBOrd}
	For any vertex in a Halin graph there exists an ordering $\oriNB$ on
	its child edges that is MSOL-definable.
\end{lemma}
\begin{proof}
	\begin{figure}[t]
		\centering
		\includegraphics[width=.45\textwidth]{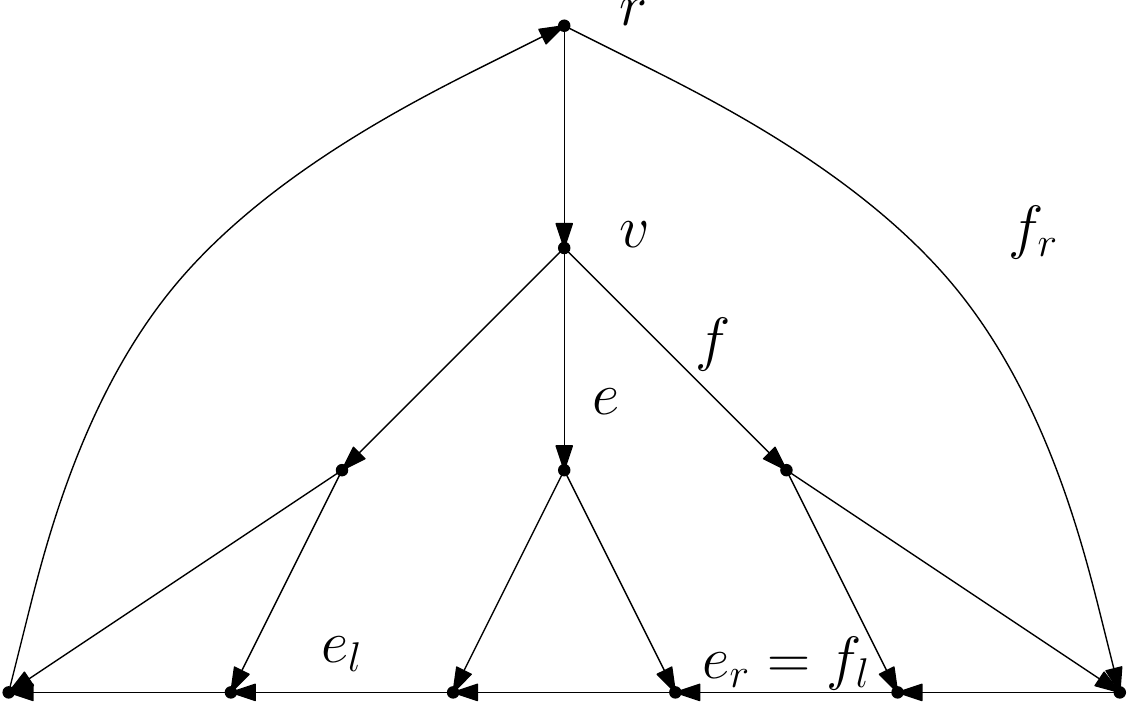}
		\caption{Example of a Halin graph with edge orientation.}
		\label{figHalinCHEdgOrd}
	\end{figure}
	Let $G = (V, E)$ be a Halin graph with an orientation on its edges as shown in
	Lemma \ref{lemHalinEdgOri}, $E_T$ its edges of the spanning tree, $E_C$ the edges of
	the outer cycle and $r$ the root of the tree $E_T$. Now, consider an
	inner vertex $v \in V$ (a non-leaf vertex w.r.t.\ the tree) and two
	child edges $e$ and $f$ of $v$ (with $e \neq f$). Every edge of a Halin graph
	is contained in exactly two fundamental cycles. Assume we have an ordering on
	the child edges of $v$ and $f$ is the right neighbor of $e$. We denote the
	edges in $E_C$, whose fundamental cycles contain $e$ and
	$f$ by $e_\ell$, $e_r$, $f_\ell$ and $f_r$, such that $e_\ell$ and $f_\ell$ ($e_r$
	and $f_r$) are contained in the left (right) fundamental cycles of $e$ and $f$,
	respectively. (See Figure \ref{figHalinCHEdgOrd} for an example.) \par
	Now consider directed paths in $E_C$ from $r$ to the tail vertices of the above
	mentioned edges. If $f$ is on the right-hand side of $e$, then the path from
	$r$ to the tail of $f_r$ is always the shortest of the four.
	The MSOL-predicates given in Appendix \ref{appSecMSOLHalinChOrd} define such
	an ordering $\oriNB(e, f)$.
\qed \end{proof}

\subsection{MSOL-Definable Tree Decompositions}\label{secTDHG}
In this section we will describe how to construct a width-3 tree decomposition
of a Halin graph that is definable in monadic second order logic. \par
First we introduce the notion of \emph{left} and \emph{right boundary} vertices
of a Halin graph with an edge orientation and ordering as described in the
previous section.
\begin{definition}[Left and Right Boundary Vertex]\label{defBoundVertHalin}
	Given a vertex $v \in V$ of a Halin graph $G$, a vertex is called its
	\emph{left boundary vertex}, denoted by $bd_l(v)$ if there exists a (possibly
	empty) path $E_P$ from $v$ to $bd_l(v)$ in $E_T$, such that the tail vertex of
	each edge in $E_P$ is the leftmost child of its parent. Similarly, we define a
	\emph{right boundary vertex} $bd_r(v)$. The \emph{boundary} of a vertex $v$ is
	the set containing both its left and right boundary vertex, denoted as $bd(v)$.
\end{definition}
Note that for all cycle vertices $v \in V_C$, we have $v = bd_l(v) = bd_r(v)$.
We now state the main result of this section.
\begin{figure}[t]
		\centering
		\subfloat[Structural overview of a Halin graph.]{
			\includegraphics[width=.6\textwidth]{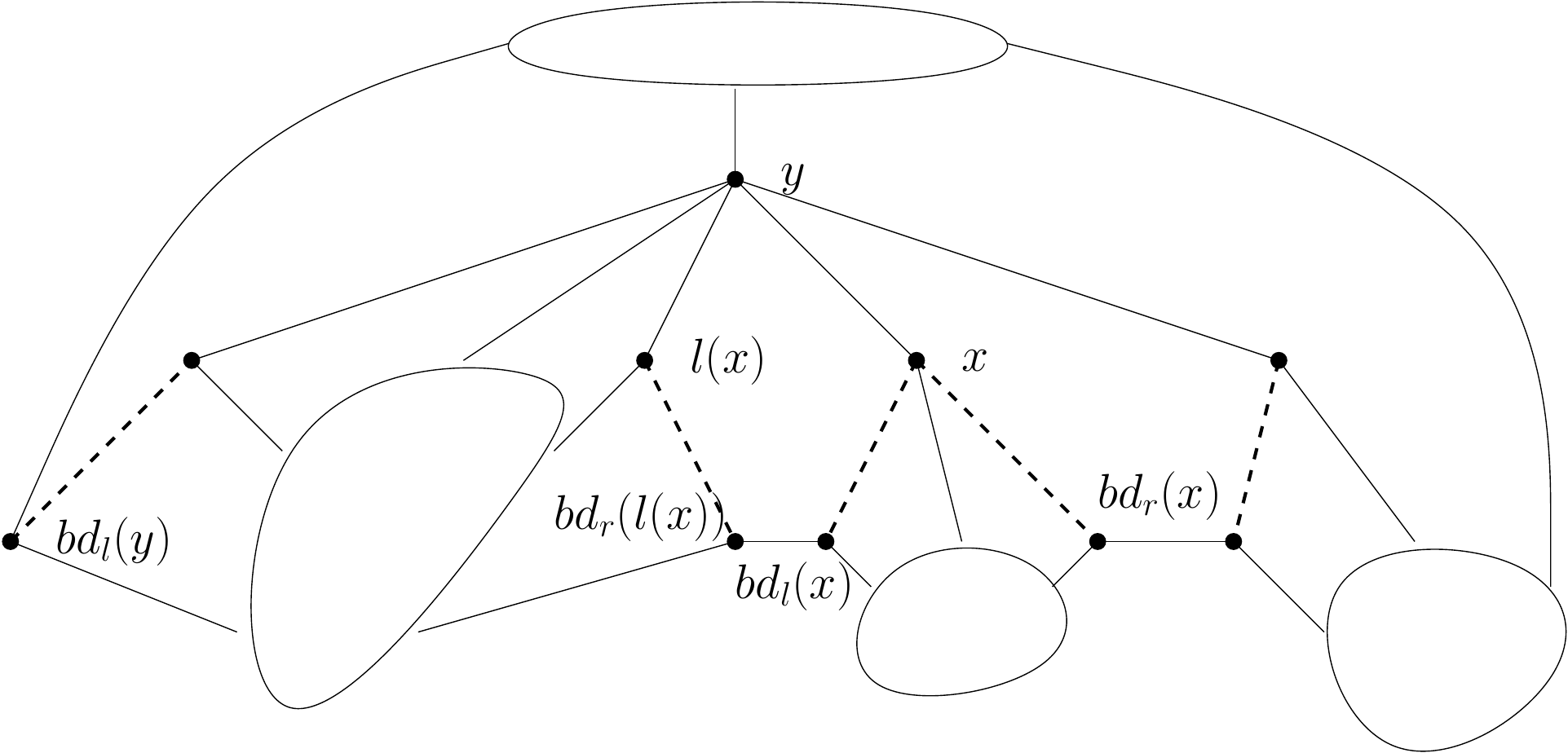}
			\label{figHalinTD1}}
 		\qquad
		\subfloat[The component created for each edge.]{
			\includegraphics[width=.235\textwidth]{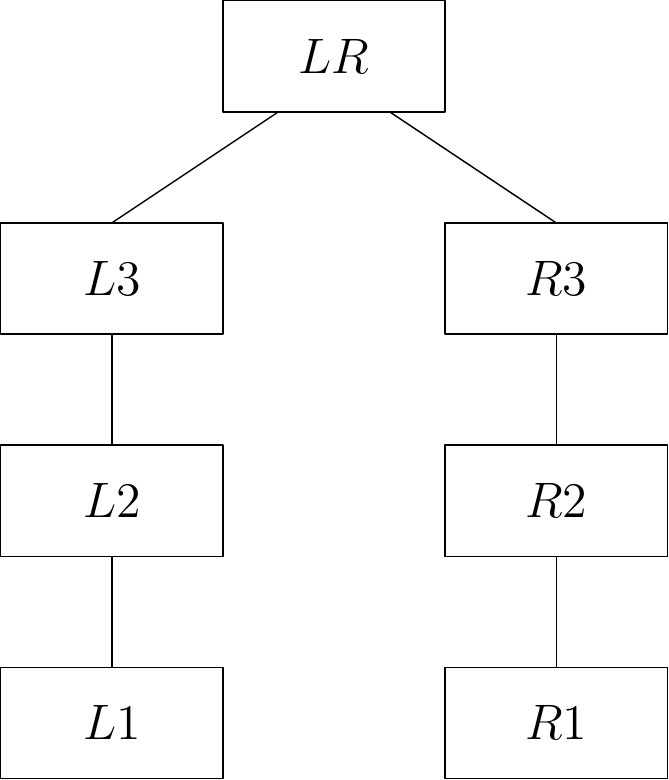}
			\label{figHalinTD2}}
		\caption{Constructing a component of a tree decomposition for an edge
 			of a Halin graph.}
 		\label{figHalinTD}
	\end{figure}
\begin{lemma}\label{lemHalinTD}
	Halin graphs admit width-3 MSOL-definable tree decompositions.
\end{lemma}
\begin{proof}
	Let $G = (V, E)$ be a Halin graph and suppose we have an orientation and
	ordering on its edges as described in Section \ref{secOrdHG}. That is, we have
	a partition $(E_C, E_T)$ of $E$ such that $E_C$ forms the (directed) outer
	cycle and $E_T$ the (directed) tree of $G$ and there is an ordering on edges
	with the same head vertex in $E_T$. \par
	For each edge $e \in E_T$ we construct a component in the tree decomposition
	that covers the edge itself and one edge on the outer cycle. A component for
	an edge $e = \{x, y\}$, where $y$ is the parent of $x$ in $E_T$ covers the
	edges $\{x, y\}$ and the edge $\{bd_r(l(x)), bd_l(x)\}$ on $E_C$, whose
	fundamental cycle both contains $\{x, y\}$ and $\{l(x), y\}$ (see Figure
	\ref{figHalinTD1} for an illustration).
	For the former we create a branch of bags of types $R1, R2$ and $R3$ and for
	the latter bags of types $L1, L2$ and $L3$, joined by a bag of type $LR$,
	containing the following vertices. \\
	\textbf{R1}. This bag contains the vertex $x$ and its boundary vertices
	$bd(x)$. \\
	\textbf{R2}. This bag contains the vertices $x$ and $y$ and the vertices
	$bd(x)$. \\
	\textbf{R3}. This bag forgets the vertex $x$ and thus contains $y$ and $bd(x)$.
	\\
	\textbf{L1}. This bag contains the vertices $y, bd_l(y)$ and $bd_r(l(x))$. \\
	\textbf{L2}. This bag introduces the vertex $bd_l(x)$ to all vertices in the
	bag $L1$. \\
	\textbf{L3}. This bag forgets the vertex $bd_r(l(x))$ and thus contains $y,
	bd_l(y)$ and $bd_l(x)$. \\
	\textbf{LR}. This bag contains the union of $L3$ and $R3$, and hence contains
	the vertices $y, bd_l(y)$ and $bd(x)$. \\
	Figure \ref{figHalinTD2} illustrates the structure of the component described
	above.
	\par
	To continue the construction, we note that removing $bd_r(x)$ from the bag of
	type $LR$ results in a bag of type $L1$ for the right neighbor edge, if such an
	edge exists.
	If $x$ is the rightmost child of $y$, then removing $bd_r(x)$ results in a bag
	of type $R1$ for the edge between $y$ and its parent in $E_T$.
	This way we can glue together components of edges using the orientation and
	ordering of the edge set of the graph. Note that if $x$ is the leftmost child
	of $y$, then it is sufficient to only create bags of types $R1, R2$ and $R3$,
	since we do not have to cover an edge on the outer cycle.
	\par
	Once we reach the root (i.e. $y$ is the root vertex
	of the graph), we only create the bags of type $R1$ and $R2$ and our
	construction is complete.
	\par
	One can verify that this construction yields a tree decomposition of $G$ and
	since the maximum number of vertices in one bag is four, its width is indeed
	three.
	\par
	To show that these tree decompositions are MSOL-definable, we note that
	we can define each bag type in MSOL in a straightforward way, once we defined
	a predicate for boundary vertices.
	The predicate $\Parent(X_p, X_c)$ requires that there are no two bags in the tree
	decomposition that contain the same vertex set and so we contract all edges
	between bags with the same vertex set. \par
	The MSOL-predicates given in Appendix	\ref{appSecMSOLHalinTD}
	complete the proof.
\qed \end{proof}
From the construction given in this proof, we can immediately derive a
consequence that will be useful in the proof of Section \ref{secFIIDHalin}.
\begin{corollary}\label{corHalinTDLS1}
	Halin graphs admit binary width-3 MSOL-definable tree decompositions such that
	all their leaf bags have size one.
\end{corollary}
\begin{proof}
	It is easy to see by the construction given in the proof of Lemma
	\ref{lemHalinTD} that this tree decomposition is binary. All leaf bags are 
	of type $R1$ and are associated with edges whose tail vertex $x$ is a
	vertex on the outer cycle. Hence, $x = bd_l(x) = bd_r(x)$ and our claim follows.
\qed \end{proof}
We will illustrate the construction of a tree decomposition given in the proof
of Lemma \ref{lemHalinTD} with the following example.
\begin{example}
	\begin{figure}
		\centering
		\subfloat[An example Halin graph.]{
			\includegraphics[width=.32\textwidth]{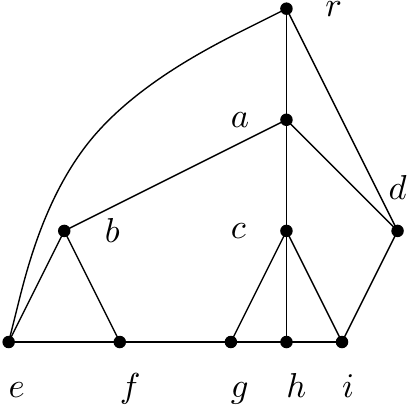}
			\label{figHalinTDExG}}
		\qquad
		\subfloat[The component of the tree decomposition corresponding to the
 			denoted edges.]{
 				\includegraphics[width=.51\textwidth]{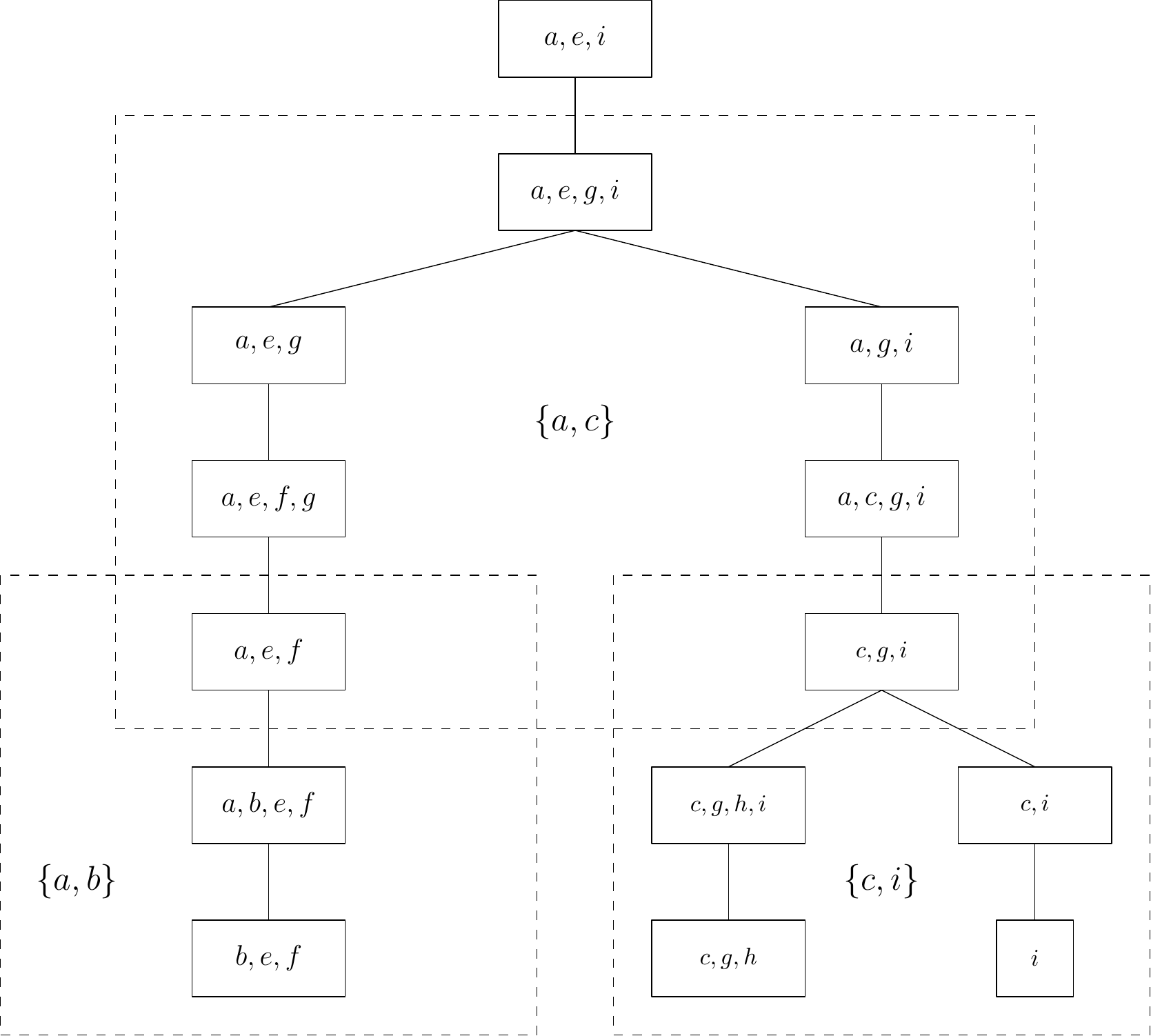}
 				\label{figHalinTDExTD}}
 		\caption{An example subtree of a tree decomposition of a Halin graph.}
 		\label{figHalinTDEx}
	\end{figure}
	Consider the graph depicted in Figure \ref{figHalinTDExG}. We are going to show
	how to create the component of its tree decomposition corresponding to the
	edges $\{a, b\}$, $\{a, c\}$ and $\{c, i\}$.
	\begin{itemize}
	  \item $\{a, b\}$: Since the vertex $b$ does not have a left sibling, we only
	  create bags $R1$, $R2$ and $R3$. Note that $LR = R3$, since $LR = L3 \cup
	  R3$, and we do not have a bag of type $L3$.
	  \item $\{c, i\}$: Since $i$ is a leaf vertex we have that $bd_l(i) = bd_r(i)
	  = i$ and so the right path starts with a bag $\{i\}$. For the same reason we
	  have that the bags $R2$ and $R3$ are equal and we contract the edge. For the
	  left path this has the effect that $L3$ and $LR$ are equal, so the
	  edge between them gets contracted as well.
	  \item $\{a, c\}$: This component can be constructed in a straightforward
	  manner. The bag $L1$ is the parent of the bag $LR$ w.r.t.\ $\{a, b\}$ and
	  $R1$ is the parent of $LR$ w.r.t.\ $\{c, i\}$. Since in both cases the vertex
	  sets are equal, we also contract these edges.
	\end{itemize}
	Figure \ref{figHalinTDExTD} shows the resulting part of the tree decomposition.
\end{example}

\subsection{Finite Index Implies MSOL-Definability}\label{secFIIDHalin}
In this section we complete the proof of our first main result, stated below. We
will also use ideas that we give here first for extending our results to
other graph classes, see Section \ref{secGen}.
\begin{lemma}\label{lemFIID}
	Finite index implies MSOL-definability for Halin graphs.
\end{lemma}
\begin{proof}
	By Lemma \ref{lemHalinTD} we know that Halin graphs admit MSOL-definable tree
	decompositions of bounded width and thus what is left to show is that we can
	define the equivalence class membership of terminal subgraphs w.r.t.\ its bags
	in monadic second order logic. \par
	We know that the graph property $P$ has finite index, so in the following we
	will denote the equivalence classes of $\sim_P$ by $C_1,\ldots,C_r$. By Lemmas
	\ref{lemEquivDerJoin} and \ref{lemEQCJoinConstCh} we know that we can derive
	the equivalence class of a terminal subgraph w.r.t.\ a node by the equivalence
	class(es) of terminal subgraphs w.r.t.\ its descendant nodes in the tree
	decomposition.
	Hence, we can conclude that the following two functions exist, also taking into account
	that our tree decomposition is binary (Corollary \ref{corHalinTDLS1}).
	\begin{proposition}\label{propEQCIndices}
		There exist two functions $f_I : \bN \times \cP(V) \to \bN$ and $f_J :
		\cP_2(\bN) \times \cP(V) \to \bN$, such that:
		\begin{enumerate}[label={(\roman*)}]
		  \item If $X$ is an intermediate bag in a tree decomposition with child bag
		  $X_c$ and $\termSG{X_c} \in C_i$, then $\termSG{X} \in C_{f_I(i, X)}$.
		  \item If $X$ is a branch bag with child bags $X_1$ and $X_2$, $\termSG{X_1}
		  \in C_i$ and $\termSG{X_2} \in C_j$, then $\termSG{X} \in C_{f(\{i, j\},
		  X)}$.\label{propEQCIndices2}
		\end{enumerate}
	\end{proposition}
	Roughly speaking, these functions can be seen as a representation of the
	transition function of an automaton that we are given in the original
	formulation of the conjecture (cf. Theorem \ref{thmMyhNerTWK}). \par
	Next, we	mimic the proof of B{\"u}chi's famous classic
	result for words over an alphabet \cite{Bue60}, as shown in \cite[Theorem
	3.1]{Tho96}.
	For each equivalence class $i$ we define sets $C_{i, \sigma}^E
	\subseteq E$ for each type $\sigma$ (see the proof of Lemma \ref{lemHalinTD})
	and equivalence class $i$. An edge $e$ is contained in set $C_{i, \sigma}^E$,
	if and only if the terminal subgraph rooted at a bag of type $\sigma$ w.r.t.\
	the edge $e$ is in equivalence class $i$.
	\par
	Our MSOL-sentence consists of three parts. First, we identify the
	equivalence classes corresponding to leaf nodes of the tree decomposition,
	and we will denote this predicate as $\phi_{Leaf}$. This is
	rather trivial, since we know that all leaf bags contain exactly one vertex
	(Corollary \ref{corHalinTDLS1}) and there is one unique equivalence class to
	which they all belong, in the following denoted by $C_{Leaf}$. Note that these
	bags are always of type $R1$.
	\par
	Second, we derive the equivalence class membership for terminal subgraphs
	using Proposition \ref{propEQCIndices}, assuming we already determined the
	equivalence class to which the terminal subgraphs w.r.t.\ its descendants
	belong.
	We denote this predicate by $\phi_{TSG}$.
	\par
	Lastly, we check if the graph corresponding to the terminal subgraph of the
	root bag of the tree decomposition is in an equivalence class satisfying $P$,
	which we denote by $\phi_{Root}$. We know that we can identify these
	equivalence classes by (the discussion given after) Theorem \ref{thmMyhNerTWK}
	and will denote them by $C_{A_1},\ldots,C_{A_p}$.
	\par
	Our MSOL-sentence then combines to:
	\begin{equation}\label{eqMSOLSentence}
		\phi_{Leaf} \wedge \phi_{TSG} \wedge \phi_{Root}
	\end{equation}
	Sentence \ref{eqMSOLSentence} together with the details for the subsentences
	given in Appendix \ref{appSecMSOLEQCMH} complete the proof.
\qed \end{proof}
Combining Lemma \ref{lemFIID} with Theorem \ref{thmMyhNerTWK} and
\cite{Cou90}, we directly obtain the following.
\begin{theorem}\label{thmDefRecHalin}
	MSOL-definability equals recognizability for Halin graphs.
\end{theorem}

\section{Extensions}\label{secGen}
The methods we used in the proofs of Section \ref{secHalin} can be
generalized and applied to a number of other graph classes, some of which we are
going to discuss in this section. The main results are presented in Sections
\ref{secTDMSOL} and \ref{secKOuterPlBD}. In the former we show that
MSOL-definability equals recognizability for any graph class that admits either
a bounded degree or an ordered MSOL-definable tree decomposition and in the
latter we give the proof for bounded degree $k$-outerplanar graphs. Furthermore
we study another subclass of $k$-outerplanar graphs in Section \ref{secKCycTr} and
graphs that can be constructed with bounded size feedback edge and vertex sets
in Section \ref{secFEFVS}. 

\subsection{MSOL-Definable Tree Decompositions}\label{secTDMSOL}
We will now turn to generalizing the proof for Halin graphs to any graph class
that admits MSOL-definable tree decompositions that are either ordered or have
bounded degree.
The proof works analogously as the proof of Lemma \ref{lemFIID}.
This result will give us a useful tool to prove Courcelle's Conjecture for a
number of graph classes, since it will follow immediately from the construction
of MSOL-definable tree decompositions.
\begin{lemma}\label{lemFIIDMTD}
	Finite index implies MSOL-definability for each graph class that admits
	MSOL-definable ordered tree decompositions of bounded width.
\end{lemma}
\begin{proof}
	It is easy to see that the predicate $\phi_{Root}$
	can be defined in the same way as in the proof of Lemma \ref{lemFIID}, only
	adding a short case analysis, since we do not necessarily know of which type
	the root bag is. Since leaf bags might not necessarily always have size one, we
	apply a small change to the tree decomposition. Assume that its width is $k$
	and that we have a $(k+1)$-coloring on the vertices of the graph, such that
	each vertex in a bag has a different color.
	Then, for each leaf bag of size greater than one, we add one child bag
	containing only the vertex with the lowest numbered color. This bag will be identified by a newly
	introduced type and associated with the same vertex/edge as its parent. We
	modify the $\Bag$- and $\Parent$-predicates accordingly and can define
	$\phi_{Leaf}$ in the same way as in Lemma \ref{lemFIID}, again including a
	case analysis as for the $\phi_{Root}$-predicate.
	\par
	Hence, in the following we only need to show how to define $\phi_{TSG}$
	to prove the claim. Again assume that the equivalence classes of $\sim_P$ are
	denoted by $C_1,\ldots,C_r$.
	We can use the function $f_I$ defined in Proposition \ref{propEQCIndices}
	to describe the relations between the equivalence classes for
	intermediate nodes. We need another function to handle
	partial terminal subgraphs w.r.t.\ a branch node, whose existence is guaranteed
	by Lemma \ref{lemEquivDerJoin}.
	\begin{proposition}
		There exists a function $f_J : \bN \times \bN \to \bN$, such that the
		following holds. If $X$ is a branch bag with child bag $Y$,
		$\pTermSG{X}{Y} \in C_i$ and $\termSG{Y} \in C_j$, then:
		\begin{enumerate}[label={(\roman*)}]
		  \item If $Y$ is the rightmost child of $X$, then $\termSG{X} \in C_{f_J(i,
		  j)}$.
		  \item Otherwise $\pTermSG{X}{r(Y)} \in C_{f_J(i, j)}$, where $\oriNBA(Y,
		  r(Y))$.
		\end{enumerate} 
	\end{proposition}
	\par 
	In the following, let $\tau \in \{\tau_1,\ldots,\tau_t\}$
	and $\sigma \in \{\sigma_1,\ldots,\sigma_s\}$. We define a number of sets, each
	one associated with an equivalence class $i$, containing either vertices or
	edges in the graph (as indicated by their upper indices), $C_{i,
	\tau}^V$ and $C_{i, \sigma}^E$. If a vertex $v$ is contained in the set
	$C_{i, \tau}^V$ this means that the terminal subgraph rooted at the bag for
	vertex $v$ of type $\tau$ is in equivalence class $i$. $C_{i, \sigma}^E$ is the
	edge set analogous to $C_{i, \tau}^V$.
	These sets can be used to define the
	equivalence class membership of terminal subgraphs rooted at intermediate
	nodes. \par
	Now let $X$ be a bag in the tree decomposition with child $Y$,
	such that the node containing $X$ is an intermediate node. We have to
	distinguish four cases when deriving the membership of a vertex/an edge in the
	respective sets, which are:
	\begin{enumerate}[label={\arabic*.}]
	  \item Both $X$ and $Y$ correspond to a vertex. \label{lemFIIDMTDCaseV}
	  \item Both $X$ and $Y$ correspond to an edge. \label{lemFIIDMTDCaseE}
	  \item $X$ corresponds to a vertex and $Y$ to an edge.
	  \label{lemFIIDMTDCaseVE}
	  \item $X$ corresponds to an edge and $Y$ to a vertex.
	  \label{lemFIIDMTDCaseEV}
	\end{enumerate}
	The predicates defining these cases for intermediate nodes are given in
	Appendix \ref{appSecMSOLEQCMGI}.
	\par
	When considering a branch node and the partial terminal subgraphs associated
	with it, we have to analyze at most eight such cases. We first turn to the
	definition of sets representing the equivalence class membership of a partial terminal
	subgraph rooted at a branch bag w.r.t. one of its children.
	Assume that a bag $X$ is of type
	$\tau$ for vertex $v$ and one of its child bags $Y$ is of type $\tau'$ for the
	vertex $v'$. Let $C_{i, \tau}^{V|P}$ and $C_{i, \tau'}^{V|C}$ be sets of
	vertices. We express that the partial terminal subgraph rooted at the bag of
	type $\tau$ for vertex $v$ w.r.t.\ the bag of type $\tau'$ for vertex $v'$ is
	in equivalence class $i$ by having $v \in C_{i, \tau}^{V|P}$ and $v' \in C_{i,
	\tau'}^{V|C}$. We define edge sets $C_{i, \sigma}^{E|P}$ and $C_{i,
	\sigma}^{E|C}$ with the same interpretation.
	The predicates for branch nodes can be found in Appendix
	\ref{appSecMSOLEQCMGJ}, which complete the proof.
\qed \end{proof}
If we are given an MSOL-definable tree decomposition that does not have an
ordering on the children of branch nodes, but instead we know that each branch
node has a constant number of children, we can prove a similar result.
\begin{lemma}\label{lemFIIDMTDBC}
	Finite index implies MSOL-definability for each graph class that admits
	bounded degree MSOL-definable tree decompositions of bounded width.
\end{lemma}
\begin{proof}
	Since this proof works almost exactly as the proof of Lemma \ref{lemFIIDMTD},
	we only state the differences. Let $c+1$ denote the maximum degree of a
	(branch) node in the tree decomposition and again we refer to the equivalence classes of
	$\sim_P$ as $C_1,\ldots,C_r$.
	Using Lemma \ref{lemEQCJoinConstCh} we know that the following holds
	(generalizing Proposition \ref{propEQCIndices}\ref{propEQCIndices2}).
	\begin{proposition}
		There exists a function $f_J : \cP_c(\bN) \times \cP(V) \to \bN$,
		such that if $X$ is a branch bag in a tree decomposition with child bags
		$X_1,\ldots,X_k$ (where $2 \le k \le c$), and each terminal subgraph
		$\termSG{X_i}$ is in equivalence class $C_{c_i}$, then the terminal subgraph
		$\termSG{X}$ is in equivalence class $f_J(\{c_1,\ldots,c_k\}, X)$.
	\end{proposition}
	Again, to define our predicate we use vertex sets $C_{i, \tau}^V$ to represent
	equivalence class membership of a terminal subgraph rooted at a vertex bag
	of type $\tau$ and edge sets $C_{i, \sigma}^E$ for edge bags of type
	$\sigma$ (and equivalence class $i$). We show how to define a predicate for
	branch bags in such tree decompositions in Appendix \ref{appSecMSOLEQCTDBD} and
	our claim follows.
\qed \end{proof}
Combining Lemmas \ref{lemFIIDMTD} and \ref{lemFIIDMTDBC} with Theorem
\ref{thmMyhNerTWK} and \cite{Cou90}, we obtain the following.
\begin{theorem}\label{thmDERSTD}
	MSOL-definability equals recognizability for graph classes that admit
	ordered or bounded degree MSOL-definable tree decompositions of width at most
	$k$.
\end{theorem}

\subsection{$k$-Cycle Trees}\label{secKCycTr}
In this section we consider graph class which can be seen as a slight
generalization of Halin graphs.
\begin{definition}[$k$-cycle trees] A graph $G$ is called \emph{cycle tree}, if
it is a planar graph that can be obtained by a planar embedding of a tree with one
	distinguished vertex $c \in V$, called the \emph{central} vertex, such that all
	vertices of distance $d$ from $c$ are connected by a cycle. If each vertex
	(except for $c$) is contained in one cycle, the number of which is $k$, then
	$G$ is called a $k$-cycle tree. We will refer to the cycle of distance $d$ from
	$c$ as the cycle $C_d$.
\end{definition}
\begin{figure}[t]
	\centering
	\subfloat[$G$ without edge orientation]{
		\includegraphics[width=.22\textwidth]{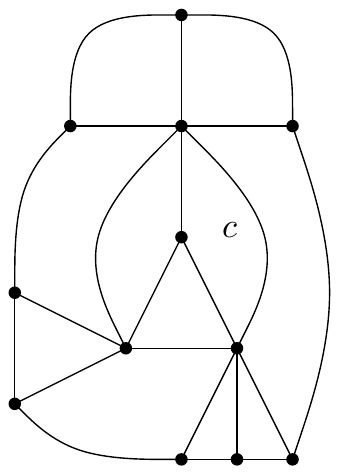}
		\label{figKLadNOri}}
	\qquad
	\subfloat[$G$ with edge orientation]{
		\includegraphics[width=.22\textwidth]{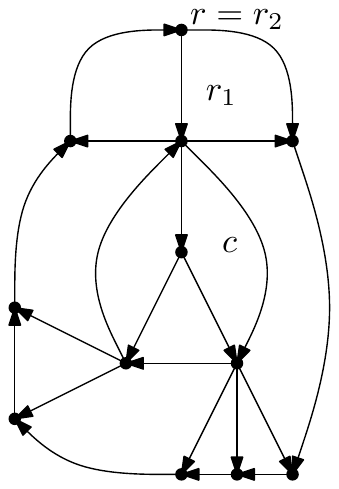}
		\label{figKLadGOri}}
	\caption{An example 2-cycle tree $G$ with central vertex $c$.}
	\label{figKLadG}
\end{figure}
Figure \ref{figKLadNOri} shows an example of a $2$-cycle tree. We easily observe
the following.
\begin{proposition}
	Each $k$-cycle tree is $k$-outerplanar.
\end{proposition}
\begin{lemma}\label{lemKOuterplOri}
	Any edge orientation $\msolOri$ using predicates $\head(e, v)$ and $\tail(e,
	v)$ is MSOL-definable for $k$-outerplanar graphs .
\end{lemma}
\begin{proof}
	This follows immediately from Lemma \ref{lemTWKOrd} and the fact that
	$k$-outerplanar graphs have treewidth at most $3k-1$ \cite[Theorem 83]{Bod98}.
\qed \end{proof}
To prove our result for $k$-cycle trees, we need the notion of the $i$-th left
and right boundary of a vertex, referring to vertices on the $i$-th cycle of the
graph.
\begin{definition}[$i$-th boundary vertex]\label{defIBound}
	Given a vertex $v$, we say that $w$ is its \emph{$i$-th left boundary vertex},
	denoted by $bd_i^l(v)$, if $w$ lies on $C_i$ and there exists a path $E_P^l$
	from $v$ to $w$, only using edges of the tree of the graph, such that no other
	path from $v$ to any vertex on $C_i$ exists that uses an edge that lies on the
	left of one of the edges in $E_P^l$. Similarly, we define the \emph{$i$-th
	right boundary vertex} $bd_i^r(v)$.
\end{definition}
Now we are ready to prove the main result of this section.
\begin{lemma}\label{lemkCycTrees}
	$k$-Cycle trees admit MSOL-definable binary tree decompositions of width
	at most $4k$.
\end{lemma}
\begin{proof}
	We can show this in almost exactly the same way as for Halin graphs (Lemma
	\ref{lemHalinTD}), so we will focus on pointing out the differences.
	Again, at first we define an edge orientation on $k$-cycle trees.
	Instead of partitioning the edge set into one directed tree and one directed
	cycle we now have one directed tree $E_T$ and $k$ directed cycles,
	such that $E_{C_i}$ denotes the cycle of distance $i$ from the central vertex
	$c$.
	\par
	The root of the tree is a vertex incident to the outermost cycle and for each
	cycle $C_i$ we have one incident root vertex $r_i$, which will be used to
	define the neighbor ordering of edges with the same head vertex.
	For a cycle $C_i$ this will be a vertex of distance $k - i$ from the root
	vertex of the tree. One can verify that this edge orientation is MSOL-definable by
	Lemma \ref{lemKOuterplOri} and the predicates given in Appendix
	\ref{appSecMSOLLG}. For an illustration of the orientation see Figure
	\ref{figKLadGOri}. \par
	Using this orientation one can define a predicate $\oriNB^i(e, f)$ for ordering
	all edges with the same parent, which then can be utilized to define $i$-th
	boundary vertices. \par
	As in the proof of Lemma \ref{lemHalinTD}, we construct a component in the
	tree decomposition for each edge $e \in E_T$.
	\begin{figure}[t]
		\centering
		\includegraphics[width=.25\textwidth]{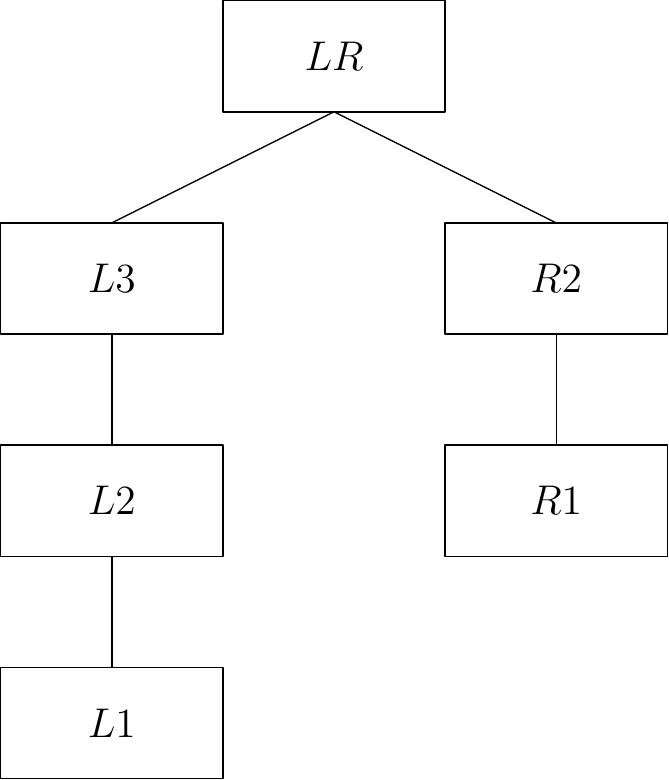}
		\caption{Bag types and edges for a component in the tree decompositions of a
		$k$-cycle tree.}
		\label{figKLadTDTypes}
	\end{figure}
	The definition of the bag types is somewhat different, since now we
	have to take into account at most $k$ cycle edges per component instead of a
	single one. Given an edge $e = \{x, y\}$ such that $y$ is the parent of $x$ and
	$y$ lies on cycle $C_i$, we have the following types of bags, with edges
	between them as shown in Figure \ref{figKLadTDTypes}.
	(Note that if in the following we refer to boundary vertices, we always mean
	the boundary vertices on higher numbered cycles.) \\
	\textbf{R1.} This bag contains the vertex $x$ and all its left and right
	boundaries. \\
	\textbf{R2.} This bag contains all vertices in the bag $R1$ plus the vertex
		$y$. \\
	\textbf{L1.} This bag contains the vertex $y$, all its left boundary vertices
	and the right boundary vertices of $y$ in the forest consisting of $E_T$
	without the edge $e$ and its right neighbors. \\
	\textbf{L2.} This bag contains all vertices of the bag $L1$ plus the left
	boundary vertices of $x$ (including $x$ itself, if $x \neq c$).
	\\
	\textbf{L3.} This bag contains the vertices of the bag $L2$ minus the right
	boundary vertices $z$ of $y$ without $e$ and its right neighbors, such that $z$
	has a matching left boundary vertex. That is, there is an edge between said
	boundary vertices and thus the vertex $z$ can be forgotten. \\
	\textbf{LR.} This bag contains the union of the bags $L3$ and $R2$.\\
	One can verify that this construction yields a tree decomposition for $k$-cycle
	trees. The largest of its bags is of type $LR$, which might contain four
	boundary sets, each of which has size at most $k$, plus the vertices $x$ and
	$y$. Since we have only one vertex, which is no boundary vertex (the central
	vertex $c$), we can conclude that the size of this bag is at most $4k + 1$ and
	hence this tree decomposition has width $4k$. The predicates in Appendix
	\ref{appSecMSOLLG} complete the proof.
\qed \end{proof}
Combining Lemma \ref{lemkCycTrees} with Theorem \ref{thmDERSTD}, we can derive
the following.
\begin{theorem}
	MSOL-definability equals recognizability for $k$-cycle trees.
\end{theorem}

\subsection{Feedback Edge and Vertex Sets}\label{secFEFVS}
In this section we consider graphs that can be obtained by the composition of a
graph that admits an MSOL-definable (ordered) tree decomposition and some
feedback edge or vertex sets, defined below.
\begin{definition}
	Let $G = (V, E)$ be a graph. An edge set $E' \subseteq E$
	is called \emph{feedback edge set}, if $G' = (V, E \setminus E')$ is acyclic.
	Analogously, a vertex set $V'$ is called \emph{feedback vertex set}, if the
	graph $G' = (V \setminus V', E \setminus E')$ is acyclic, where $E'$ denotes
	the set of incident edges of $V'$ in $E$.
\end{definition}
\begin{theorem}\label{thmFESFVS}
	Let $G = (V, E)$ be a graph with spanning tree $T = (V, F)$, which
	admits an MSOL-definable (ordered) tree decomposition of width $k$, such that its vertex
	and edge bag predicates are associated with either (a subset of the) vertices
	of the graph or (a subset of the) edges in the spanning tree. \par
	Let $l$ be a constant. A graph $G'$ admits an MSOL-definable (ordered)
	tree decomposition of width $k + l$, if one of the following holds.
	\begin{enumerate}[label={(\roman*)}]
	  \item Let $E'$ denote a set of edges, such that each biconnected component of
	  the graph $T' = (V, F \cup E')$ has a feedback edge set of size at most
	  $l$, where $G' = (V, E \cup E')$.\label{thmFESFVSE}
	  \item Let $V'$ denote a set of vertices and $E' \subseteq (V \times V') \cup
	  (V' \times V')$ a set of incident edges, such that each biconnected component
	  of the graph $T' = (V \cup V', F \cup E')$ has a feedback vertex set of size
	  at most $l$, where $G' = (V \cup V', E \cup E')$.\label{thmFESFVSV}
	\end{enumerate}
\end{theorem}
\begin{proof}
	\ref{thmFESFVSE}. Let $e = \{v, w\}$ be an edge in $E'$ and note that since $G$
	has bounded treewidth $k$, there exists
	a $(k+1)$-coloring on its vertices. 
	Assume wlog.\ that the coloring set is a set of natural
	numbers $\{1,\ldots,k+1\}$ and $\col(v) < \col(w)$. Then we add the vertex
	$v$ to each bag that is associated with either a vertex or an edge in $T$ that
	lie on the fundamental cycle of $e$. The width of the tree decomposition
	increased by at most $l$ (by Lemmas 6 and 73 in \cite{Bod98}).
	\par
	\ref{thmFESFVSV}. Let $v$ be a vertex in $V'$. We add
	$v$ to all bags that correspond to vertices/edges contained in the same
	biconnected component as $v$ (in $T'$). The fact that the treewidth increased
	by at most $l$ follows from \cite[Lemmas 6 and 72]{Bod98}.
	\par
	In Appendix \ref{appSecMSOLFESFVS} we show how to extend all predicates to
	include the newly introduced vertices in the bags for both cases.
\qed \end{proof}
As an example we apply Theorem \ref{thmFESFVS} to both Halin graphs and
$k$-cycle trees, which - in combination with Theorem \ref{thmDERSTD} - yields
the following result.
\begin{theorem}
	Let $\cC$ denote a graph class such that its members can be constructed from a
	Halin graph or a $k$-cycle tree together with either an edge set or vertex set
	as described in Theorem \ref{thmFESFVS}. Then, MSOL-definability equals
	recognizability for all members of $\cC$.
\end{theorem}

\subsection{Bounded Degree $k$-Outerplanar Graphs}\label{secKOuterPlBD}
We now give another method for proving Courcelle's conjecture based on the
notion of vertex and edge remember numbers, which will enable us to prove it for
$k$-outerplanar graphs of bounded degree. We first give the necessary
definitions.
\begin{definition}[Vertex and Edge Remember Number]
	Let $G = (V, E)$ be a graph with maximal spanning forest $T = (V, F)$. The
	\emph{vertex remember number} of $G$ (with respect to $T$), denoted by
	$vr(G, T)$, is the maximum number over all vertices $v \in V$ of fundamental
	cycles that use $v$. Analogously, we define the \emph{edge remember number}, denoted by
	$er(G, T)$.
\end{definition}
\begin{theorem}\label{thmVRER}
	Let $G = (V, E)$ be a graph with a spanning tree $T =
	(V, F)$ and let $k = \max\{vr(G, T), er(G, T) + 1\}$.
	$G$ admits \begin{enumerate}[label={(\roman*)}]
	  \item a width-$k$ MSOL-definable tree decomposition of bounded degree, if $G$
	  has bounded degree.\label{thmVRERBD}
	  \item a width-$k$ MSOL-definable ordered tree decomposition, if there is
	  an MSOL-definable ordering $\oriNB(e, f)$ over all edges $e, f \in F$ with
	  the same head vertex.\label{thmVREROri}
	\end{enumerate}
\end{theorem}
\begin{proof}
	For both \ref{thmVRERBD} and \ref{thmVREROri} we can construct a tree
	decomposition $(T', X)$ as shown in the proof of Theorem 71 in \cite{Bod98}.
	That is, we create a tree $T' = (V \cup F, F')$, where $F' = \{\{v, e\} \mid v
	\in V, e \in F, \exists w \in V : e = \{v, w\}\}$, i.e. we add an extra node
	between each two adjacent vertices in the spanning tree. The construction of
	the sets $X_{t}, t \in V \cup F$ works as follows. For a bag associated with a
	vertex $v$ in the spanning tree we first add $v$ to $X_v$, and for a bag
	associated with an edge $e$, we add both its endpoints to $X_e$. Then, for each
	edge $e \in E \setminus F$, we add one of its endpoints to each bag
	corresponding to a vertex or edge on the fundamental cycle of $e$.
	To make sure that our method of choosing one endpoint of an edge is
	MSOL-definable, we use the same argument as in the proof of Theorem
	\ref{thmFESFVS}\ref{thmFESFVSE}. That is, we assume the existence of a vertex
	coloring in the graph and pick the vertex with the lower numbered color. \par
	One can verify that $(T', X)$ is a tree decomposition of $G$ and we have for
	all vertex bags $X_v$ that $|X_v| \le 1 + vr(G, T)$ and for all edge bags $X_e$
	that $|X_e| \le 2 + er(G, T)$ and thus the claimed width of $(T', X)$ follows.
	\par
	Now we show that finding a spanning tree such that its vertex
	and edge remember number are bounded by a constant, say $\kappa$, is
	MSOL-definable, if it exists. We can simply do this by guessing an edge set
	$E_T \subseteq E$ and checking whether $E_T$ is the edge set of a spanning tree
	in $G$ with the claimed bound on the resulting vertex and edge remember
	numbers. Since $\kappa$ is constant, this can be done in a straightforward way,
	see Appendix \ref{appSecMSOLVRER}.
	\par
	For defining the $\Bag$- and $\Parent$-predicates, we assume wlog.\ that we
	have a root and an MSOL-definable orientation on the edges
	in the spanning tree,\footnote{This clearly holds by Lemma \ref{lemTWKOrd},
	since trees have treewidth 1.} so we can directly define such predicates,
	see Appendix \ref{appSecMSOLVRER}. \par
	For case \ref{thmVRERBD} one easily sees that $(T', X)$ has bounded degree,
	since the degree of any node corresponding to a vertex $v \in V$
	in the tree decomposition is equal to the degree of $v$ in $G$. Nodes
	containing edge bags are always intermediate nodes.
	\par
	Case \ref{thmVREROri} holds, since we can define an orientation $\oriNB(X_a,
	X_b)$ for the children of each vertex bag by using the ordering of its
	corresponding edges. \par
	The predicates defined in Appendix \ref{appSecMSOLVRER} complete the proof.
\qed \end{proof}
In his proof for the treewidth of $k$-outerplanar graphs being $3k-1$,
Bodlaender used the following lemma.
\begin{lemma}[Lemma 81 in \cite{Bod98}]
	Let $G = (V, E)$ be a $k$-outerplanar graph with maximum degree 3. Then there
	exists a maximal spanning forest $T = (V, F)$ with $er(G, T) \le 2k$ and
	$vr(G, T) \le 3k - 1$.
\end{lemma}
Given the nature of its proof, one immediately has the following
consequence.
\begin{corollary}\label{corKOuterplSFBD}
	Let $G = (V, E)$ be a $k$-outerplanar graph with maximum degree $\Delta$. Then
	there exists a maximal spanning forest $T = (V, F)$ with $er(G, T) \le 2k$ and
	$vr(G, T) \le \Delta k - 1$.
\end{corollary}
We can now prove the main result of this section.
\begin{theorem}
	MSOL-definability equals recognizability for $k$-outerplanar graphs of bounded
	degree.
\end{theorem}
\begin{proof}
	Let $G = (V, E)$ be a $k$-outerplanar graph with maximum degree $\Delta$.
	By Corollary \ref{corKOuterplSFBD}, we know that there exists a maximal
	spanning forest $T = (V, F)$ of $G$ with $er(G, T) \le 2k$ and $vr(G, T) \le
	\Delta k-1$. By Theorem \ref{thmVRER}\ref{thmVRERBD}, we know that $G$ admits
	an MSOL-definable tree decomposition of bounded degree. If $\Delta < 3$, then
	the width of this tree decomposition is at most $4k+1$, and if $\Delta \ge 3$,
	it is at most $\Delta k - 1$, so in both cases the width is bounded by a
	constant. The rest now follows from Theorem \ref{thmDERSTD}.
\qed \end{proof}
Note that the theorem also holds, if we add feedback edge and vertex sets to a
$k$-outerplanar graph of bounded degree, as explained in Theorem \ref{thmFESFVS}.

\section{Conclusion}\label{secConc}
In this paper we showed that MSOL-definability equals recognizability for Halin
graphs, $k$-cycle trees, graph classes constructed using certain feedback
edge or vertex sets and bounded degree $k$-outerplanar graphs.
Hence we proved a number of special cases of Courcelle's Conjecture
\cite{Cou90}, which states that each graph property that is recognizable for graphs of
bounded treewidth is CMSOL-definable, additionally
strengthening it to MSOL-definability. \par
For our proofs, we introduced the concept of MSOL-definable tree decompositions,
and used MSOL-definable tree decompositions of bounded degree or ordered
MSOL-definable tree decompositions (i.e.\ admitting an ordering
on nodes with the same parent). We additionally showed that this conjecture
holds for any graph class that admits either one of these kinds of tree decompositions. \par
We hope that the techniques of our paper give useful tools to solve other
special cases in the future, and also help to establish the border between cases
that allow MSOL-definability versus cases that need the counting predicate of
CMSOL.
\par
We plan to further investigate the case of $k$-outerplanar graphs and believe
that the following conjecture holds. 
\begin{conjecture} Recognizability equals	
	\begin{enumerate}[label={(\roman*)}]
	  \item MSOL-definability for 3-connected $k$-outerplanar graphs.
	  \item CMSOL-definability for $k$-outerplanar graphs.
	 \end{enumerate}
\end{conjecture}
We also hope to establish that 3-connectedness is a necessary condition to avoid
the counting predicate in our proof, which for $k$-outerplanar graphs will
provide us a with clear separation between MSOL and CMSOL. \par
Another interesting graph property that might be used in such proofs is
Hamiltonicity (in our sense that means a graph admits a Hamiltonian
\emph{path}).
It is easy to see that one can order nodes with the same parent in
an MSOL-definable tree decomposition, if the underlying graph admits a
Hamiltonian path, hence we conjecture the following.
\begin{conjecture}
	MSOL-definability equals recognizability for (3-connected) Hamiltonian partial
	$k$-trees.
\end{conjecture}

\subsection*{Acknowledgements}
The second author thanks Bruno Courcelle, Mike Fellows, Pinar Heggernes and Jan Arne 
Telle for inspiring discussions.

\bibliography{References}
\newpage
\appendix

\section{Monadic Second Order Predicates and Sentences}\label{appSecMSOL}
We build sentences in monadic second order logic from a collection
of predicates. Once we defined these predicates they will be the building blocks
of more complex expressions, joined by MSOL-connectives and/or quantification of
its declared variables. Hence, we follow the ideas of the work of Borie et al.
\cite{BPT92}, who also give a large list of predicates and their definitions. \\
Note that the length of our sentences and formulas always has to be bounded by
some constant, independent of the size of the input graph. \par
We will denote single element variables by small letters, where $v, w, v', w',
\ldots$ typically represent vertices and $e, f, e', f',\ldots$ edges. Set
variables will be denoted by capital letters. Unless stated otherwise
explicitly, $V$ always denotes the vertex set of some input graph $G$ and $E$
its edge set. Since we always assume our predicates to appear in the context of
such a graph we might drop these two variables as an argument of a predicate.
\par
By some trivial definition, the following predicates are MSOL-definable (see
also Theorem 1 in \cite{BPT92}). In our text we might refer to them as the
\emph{atomic} predicates of monadic second order logic over graphs.
\begin{enumerate}[label={(\Roman*)}]
  \item $v = w$ (Vertex equality)
  \item $\Inc(e, v)$ (Vertex-edge incidence)
  \item $v \in V$ (Vertex membership)
  \item $e \in E$ (Edge membership)
\end{enumerate}
Note that to shorten our notation we might omit statements such as $v \in V$ or
$e \in E$ when quantifying over a variable. In this case we are referring to some
vertex/edge in the whole graph and the interpretation of the variables will
always be obvious from the context or the notational conventions explained above. \par
From the atomic predicates, one can directly derive the following:
\begin{itemize}
  \item $\Adj(v, w, E)$ (Adjacency of $v$ and $w$ in $E$)
  \item $\Edge(e, v, w)$ ($e = \{v, w\}$)
\end{itemize}
In a straightforward way (and by Theorem 4 in \cite{BPT92}), one can see that
the following are MSOL-definable:
\begin{itemize}
  \item $V = V' \cup V''$, $V = V' \setminus V''$, $V = V' \cap V''$ (plus the
  edge set equivalents)
  \item $V' = \IncV(E')$ [$E' = \IncE(V')$] ($V'$ [$E'$] is the set of incident
  vertices [edges] of $E'$ [$V'$])
  \item $\deg(v, E) = k$ ($v$ has degree $k$ in $E$, where $k$ is a constant)
  \item $\Conn(V, E)$, $\Conn_k(V, E)$, $\Cycle(V, E)$, $\Tree(V, E)$, $\Path(V,
  E)$
\end{itemize}

\subsection{Edge Orientation of a Halin Graph}\label{appSecMSOLHalinOrd}
In the current section we show how to define an edge orientation on a Halin
graph as explained in the proof of Lemma \ref{lemHalinEdgOri}. That is, we will
define a partition of the edge set of the graph into a directed tree $E_T$ and a
directed cycle $E_C$. \par 
As outlined in the proof, we use a coloring on its vertex set to define the
orientation of edges. Since we will use this result in later sections as well,
we define the general case of a $k$-coloring on the vertices of a graph.
\begin{align*}
	\Part_V(V, X_1,\ldots,X_k) \Leftrightarrow & (\forall v \in V)\Big(
		\bigvee_{1 \le i \le k} v \in X_i \wedge \bigwedge_{\stackrel{1 \le i
		\le k}{j \neq i}} \neg v \in X_j\Big) \\
	\kCol{k}(X_1,\ldots,X_k) \Leftrightarrow & \Part_V(V, X_1,\ldots,X_k) \\
	 \wedge&\forall e\forall v\forall w \Big(\Edge(e, v, w) \to \bigwedge_{1 \le i
	 \le k} \neg (v \in X_i \wedge w \in X_i) \Big)
\end{align*}
Now we define a predicate $\head(e, v)$ that is true if and only if $v$ is the
head vertex of the edge $e$ in the given orientation by comparing the indices
of the color classes that contain an endpoint of $e$. Note that the following
predicates always appear in the scope of an edge set $F$ and a $k$-coloring
$X_1,\ldots,X_k$.
\begin{align*}
	\col_<(v, w) \Leftrightarrow & \bigvee_{1 \le i < j \le k}(v \in X_i \wedge w
		\in X_j) \\
	\head(e, v) \Leftrightarrow &\exists w (\Edge(e, v, w) \wedge e \in F
	\leftrightarrow \col_<(v, w)) \\
	\tail(e, v) \Leftrightarrow &\exists w (\Edge(e, v, w) \wedge \neg e \in F
	\leftrightarrow \col_<(v, w)) \\
	\Arc(e, v, w) \Leftrightarrow &\Edge(e, v, w) \wedge \head(e, v) ~~~[e = (v,
	w)]
\end{align*}
Analogously to the definition of vertex degree predicates $\deg(v, E)$, as shown
in \cite[Theorem 4]{BPT92}, we can define predicates $\indeg(v, E)$ and
$\outdeg(v, E)$ for the in-degree and out-degree of a vertex in a directed
graph. We show how to define that the in-degree of a vertex is equal to a certain
constant $k$.
\begin{align*}
	\indeg(v, E) \ge k \Leftrightarrow &\exists w_1\cdots\exists w_k
		\Big(\Big(\bigwedge_{1 \le i \le k} (\exists e \in E) \Arc(e, w_i, v) \Big) \\
		&\wedge \bigwedge_{1 \le i < j \le k} \neg w_i = w_j\Big) \\
	\indeg(v, E) \le k \Leftrightarrow &\forall w_1 \cdots \forall w_{k+1} 
		\Big(\Big(\bigwedge_{1 \le i \le k+1} (\exists e \in E) \Arc(e, w_i, v) \Big)
		\\ 
		&\to \bigvee_{1 \le i < j \le k+1} w_i = w_j \Big) \\
	\indeg(v, E) = k \Leftrightarrow &\indeg(v, E) \le k \wedge \indeg(v, E) \ge k
\end{align*}
In a similar way we can define predicates for the out-degree and regularity
of a vertex for in- and out-degree and both (denoted by $\inKReg{k}$,
$\outKReg{k}$ and $\inOutKReg{k}$, respectively).
This enables us to define predicates for directed trees and cycles.
\begin{align*}
	\dir{\Cycle}(V, E) \Leftrightarrow &\Conn(V, E) \wedge \inOutKReg{1}(V, E) \\
	\dir{\Tree}(V, E) \Leftrightarrow &\Tree(V, E) \wedge (\exists r \in V)
	(\forall v \in V) \Big((r = v \wedge \indeg(v, E) = 0) \\
	&\vee (\neg v = r \wedge \indeg(v, E) = 1)\Big)
\end{align*}

\subsection{Child Ordering of a Halin Graph}\label{appSecMSOLHalinChOrd}
This section concludes the proof of Lemma \ref{lemHalinNBOrd}, that is we define
an ordering on edges in a Halin graph that have the same parent in the
tree $E_T$. Therefor we define predicates for
directed paths and fundamental cycles. Note that $\dir{\Path}(s, t, E')$ is true
if and only if $E'$ is a directed $s-t$-path.
\begin{align*}
	\dir{\Path}(V, E) \Leftrightarrow &\dir{\Tree}(V, E) \wedge (\forall v \in V)
	\deg(v, E) \le 2 \\
	\dir{\Path}(s, t, E') \Leftrightarrow & \dir{\Path}(\IncV(E'), E') \wedge
	\indeg(s) = 0 \wedge \outdeg(t) = 0
\end{align*}
Now we turn to the notion of fundamental cycles. 
We assume that the following predicates appear within the
scope of an edge set $E_T$, which is a spanning tree of the given graph.
\begin{align*}
	\FundCyc(E') \Leftrightarrow &\Cycle(\IncV(E'), E') \wedge (\exists e \in
	E')(\forall e' \in E')(\neg(e = e') \leftrightarrow e \in E_T) \\
	\FundCyc(e, e') \Leftrightarrow &(\exists E' \subseteq E)(e \in E'
	\wedge e' \in E' \wedge \FundCyc(E'))
\end{align*}
Note that $\FundCyc(e, e')$ is true if and only if there exists a fundamental
cycle in the graph containing both $e$ and $e'$. Now we can define an ordering
$\oriNB(e, f)$ on edges with the same parent, as explained in the proof of Lemma
\ref{lemHalinNBOrd}.
\begin{align*}
	\oriNB(e, f) \Leftrightarrow & \head(e) = \head(f) \wedge (\exists f' \in
	E_C)(\forall e' \in E_C)(\forall F' \subseteq E_C)(\forall E' \subseteq E_C) \\
		\Big(\Big( &\FundCyc(e, e') \wedge \FundCyc(f, f') \wedge \dir{\Path}(r,
		\tail(e'), E') \\
		 &\wedge \dir{\Path}(r, \tail(f'), F') \Big) \to F' \subset E'
		\Big)
\end{align*}
Furthermore we define a predicate $\oriNBA(e, f)$ that is true if and only if
$f$ is the leftmost right neighbor of $e$ and vice versa. We also apply this
notion to vertex variables, which allows us to refer to left and right siblings
of a vertex. We denote these predicates by $\oriSIB(x, y)$ and $\oriSIBA(x, y)$.
\begin{align*}
	\oriNBA(e, f) \Leftrightarrow & \oriNB(e, f) \wedge \forall f'((\neg f = f'
	\wedge \oriNB(e, f')) \to \oriNB(f, f')) \\
	\oriSIB(x, y) \Leftrightarrow &\exists e \exists f (\tail(e, x) \wedge
	\tail(f, y) \wedge \oriNB(e, f)) \\ 
	\oriSIBA(x, y) \Leftrightarrow & \exists e\exists f (\tail(e, x)
	\wedge \tail(f, y) \wedge \oriNBA(e, f))
\end{align*}
In the following we will use the rewrite of $\oriSIBA$ to
\begin{equation*}
	y = l(x) \Leftrightarrow \oriSIBA(y, x).
\end{equation*}
This expresses that a vertex $y$ is the direct left sibling of the vertex $x$ in
our ordering.

\subsection{Tree Decomposition of a Halin Graph}\label{appSecMSOLHalinTD}
In this section we define predicates $\Bag_\sigma(e, X)$ for all bag types used
in the proof of Lemma \ref{lemHalinTD}, and $\Parent(X_p, X_c)$ according to
the given construction. 
In the following we assume that we are given an edge $e \in E_T$, $e = \{x,
y\}$, such that $y$ is the parent of $x$ in $E_T$.

\subsubsection{Boundary vertices}
For defining predicates for bag types in our tree decomposition, we need to
show how to define boundary vertices in MSOL. First, we define
predicates to check whether a vertex is the right-(/left-)most child of its
parent.
\begin{align*}
	\Child_{R+}(x) \Leftrightarrow &\forall y\forall z \forall e \forall e'
	((\Arc(e, y, x) \wedge \Arc(e', y, z)) \to \oriNB(e', e))
\end{align*}
Note that $\Child_{L+}(x)$ can be defined similarly, replacing $\oriNB(e',
e)$ by $\oriNB(e, e')$. In the following we let $V_C = \IncV(E_C)$.
\begin{align*}
	y = bd_r(x) \Leftrightarrow &(x \in V_C \wedge x = y) \vee \Big(x \in V
	\wedge y \in V_C \\
	\wedge &\Big((\exists E_P \subseteq E_T)(\dir{\Path}(x, y, E_P) \wedge (\forall
	e \in E_P) \\
	&(\forall z (\tail(e, z) \to \Child_{R+}(z))))\Big)\Big)
\end{align*}
Replaying $\Child_{R+}$ by $\Child_{L+}$ in the above predicate we can also
define $y = bd_l(x)$.

\subsubsection{Bag Types}
We define an MSOL-predicate for each bag type that we introduced in the proof of
Lemma \ref{lemHalinTD}. Using the definition of boundary vertices given above, we can
define them in a straightforward manner.
\begin{align*}
	\Bag_{R1}(e, X) \Leftrightarrow & (x' \in X) \leftrightarrow (x' = x \vee x'
		= bd_r(x) \vee x' = bd_l(x)) \\
	\Bag_{R2}(e, X) \Leftrightarrow & (x' \in X) \leftrightarrow (x' = y \vee x'
	= x \vee x' = bd_r(x) \vee x' = bd_l(x)) \\
	\Bag_{R3}(e, X) \Leftrightarrow & (x' \in X) \leftrightarrow (x' = y \vee x' =
	bd_r(x) \vee x' = bd_l(x)) \\
	\Bag_{L1}(e, X) \Leftrightarrow & (x' \in X) \leftrightarrow (x' = y \vee x' =
	bd_l(y) \vee bd_r(l(x))) \\
	\Bag_{L2}(e, X) \Leftrightarrow & (x' \in X) \leftrightarrow (x' = y \vee x' =
	bd_l(y) \vee x' = bd_r(l(x)) \vee x' = bd_l(x)) \\
	\Bag_{L3}(e, X) \Leftrightarrow & (x' \in X) \leftrightarrow (x' = y \vee x' =
	bd_l(y) \vee x' = bd_l(x)) \\
	\Bag_{LR}(e, X) \Leftrightarrow & (x' \in X) \leftrightarrow (x' = y \vee x' =
	bd_l(y) \vee x' = bd_r(x) \vee x' = bd_l(x))
\end{align*}
As a next step we will unify the above predicates, to deal with the
cases when certain bags do not need to be created for an edge. This is the case
when we reach the root vertex of the graph or whenever an edge is the leftmost
child edge of a vertex.
\begin{align*}
	\Bag(X) \Leftrightarrow &\exists e \Big(y = r \wedge (\Bag_{R1}(e, X) \vee
	\Bag_{R2}(e, X)) \\
	&\vee\Big(\neg y = r \wedge \Big((\Child_{L+}(x) \wedge (\Bag_{R1}(e, X) \vee
	\Bag_{R2}(e, X) \\
	&\vee \Bag_{R3}(e, X))) \vee (\neg \Child_{L+}(x) \wedge (\Bag_{R1}(e, X) 
	\\ &\vee \cdots \vee \Bag_{LR}(e, X)))\Big)\Big)\Big)
\end{align*}

\subsubsection{The Parent Relation}
We now turn to defining the predicate $\Parent(X_p, X_c)$,
which is true if and only if the bag $X_p$ is the parent bag of $X_c$ in the
tree decomposition. Due to the
contraction step we can only have edges between bags if
their vertex sets are not equal. Note that adding the term '$\neg X_p = X_c$' is
sufficient to represent these contractions.
The rest is a case analysis as implied by Figure
\ref{figHalinTD2} and the respective parent/child relationships between
components.
\begin{align*}
	\Parent(X_p, X_c) \Leftrightarrow &\Bag(X_p) \wedge \Bag(X_c) \wedge \neg X_p =
	X_c \wedge (\Parent_I(X_p, X_c) \\ &\vee\Parent_{NB}(X_p, X_c) \vee
	\Parent_P(X_p, X_c))
	\\
	\Parent_I(X_p, X_c) \Leftrightarrow & \exists e
	\Big((\Bag_{R1}(e, X_c) \wedge \Bag_{R2}(e, X_p)) \\
	&\vee (\Bag_{R2}(e, X_c) \wedge \Bag_{R3}(e, X_p)) \\
	&\vee ((\Bag_{R3}(e, X_c) \vee \Bag_{L3}(e, X_c)) \wedge \Bag_{LR}(e, X_p)) \\
	&\vee (\Bag_{L1}(e, X_c) \wedge \Bag_{L2}(e, X_p)) \\
	&\vee (\Bag_{L2}(e, X_c) \wedge \Bag_{L3}(e, X_p)) 
	\Big) \\
	\Parent_{NB}(X_p, X_c) \Leftrightarrow &\exists e \exists e' (\oriNBA(e, e')
	\wedge \Bag_{LR}(e, X_c) \wedge \Bag_{L1}(e', X_p)) \\
	\Parent_{P}(X_p, X_c) \Leftrightarrow &\exists e \exists e' (\Child_{R+}(x)
		\wedge \tail(e', y) \\
		&\wedge \Bag_{LR}(e, X_c) \wedge \Bag_{R1}(e', X_p))
\end{align*}

\subsection{Equivalence Class Membership for Halin
Graphs}\label{appSecMSOLEQCMH} 
In this section we complete the proof of Lemma \ref{lemFIID}, which states that
finite index implies MSOL-definability for Halin graphs. In particular we define
the predicates $\phi_{Leaf}$, $\phi_{TSG}$ and $\phi_{Root}$, which represent
the cases for leaf bags, inner bags (i.e., intermediate and branch bags that are
not the root) and the root bag, respectively. \par
The predicate $\phi_{Leaf}$ can be defined in a straightforward way, using
the fact that we know that all terminal subgraphs of leaf bags are in the
equivalence class $C_{Leaf}$ and that leaf bags are always of type $R1$.
\begin{align*}
	\phi_{Leaf} = \forall X\forall e((\Bag_{R1}(e, X) \wedge \Leaf(X)) \to e \in
	C_{Leaf, R1})
\end{align*}
Next, we turn to defining $\phi_{TSG}$, where we distinguish two cases. That is,
either $X$ is an intermediate or a branch bag. We conduct the case analysis as
implied by the construction of our tree decomposition as shown in Section
\ref{secTDHG}.
\begin{align*}
	\phi_{TSG} = &\Big(\exists C_{i, L1} \exists C_{i, L2} \exists
	C_{i, L3} \exists C_{i, R1} \exists C_{i, R2} \exists C_{i, R3} \exists C_{i,
	LR}\Big)_{i = 1,\ldots,r} \\
	&\forall X \forall Y \Big((\Parent(X, Y) \wedge \Intermediate(X)) \to
	\phi_{TSG, Int} \\
	&\wedge \forall Y'(\neg(Y = Y') \wedge \Parent(X, Y) \wedge \Parent(X, Y')
	\wedge \JoinB(X)) \\
	&\to \phi_{TSG, Branch} \Big)
\end{align*}
The first case we are considering is when $X$ is an intermediate node with
child bag $Y$. These edges either belong to the same component, which is handled
in the first part of the predicate, or they belong to components of different
edges, such that the two are either direct neighbor edges according to the
$\oriNBA$-ordering or one of the edges is the parent edge of the other one.
\begin{align*}
	\phi_{TSG, Int} = &\forall e \Big(
	(\Bag_{L2}(e, X) \wedge \Bag_{L1}(e, Y)) \to \bigwedge_{i =
	1,\ldots,r}(e \in C_{i, L1} \to e \in C_{f_I(i, X), L2}) \\
	&\vee (\Bag_{L3}(e, X) \wedge \Bag_{L2}(e, Y)) \to \bigwedge_{i = 1,\ldots,r}
	(e \in C_{i, L2} \to e \in C_{f_I(i, X), L3})  \\
	&\vee (\Bag_{R2}(e, X) \wedge \Bag_{R1}(e, Y)) \to \bigwedge_{i = 1,\ldots,r}
	(e \in C_{i, R1} \to e \in C_{f_I(i, X), R2}) \\
	&\vee (\Bag_{R3}(e, X) \wedge \Bag_{R2}(e, Y)) \to
	\bigwedge_{i=1,\ldots,r} (e \in C_{i, R2} \to e \in C_{f_I(i,
	X), R3})\Big) \\
	&\vee \forall e \forall e'\Big((\Parent_{NB}(X, Y) \wedge \Bag_{L1}(e', X)
	\wedge \Bag_{LR}(e, Y))
	\\
	 &\to \bigwedge_{i = 1,\ldots,r} (e \in C_{i, LR} \to e' \in
	 C_{f_I(i, X), L1}) \Big) \\
	&\vee \Big((\Parent_P(X, Y) \wedge \Bag_{R1}(e', X) \wedge \Bag_{LR}(e, Y)) \\
	&\to \bigwedge_{i=1,\ldots,r} (e \in C_{i, LR} \to e' \in
	C_{f_I(i ,X), R1}) \Big)
\end{align*}
Now we assume that $X$ is a branch node with child bags $Y$ and $Y'$. We can't
identify the types of the bags $Y$ and $Y'$ immediately, since some of the edges
in the component might have been contracted. So in the following, let $L$ denote
the type $L1, L2$ or $L3$, and $R$, respectively, $R1, R2$ or $R3$. We can
define each combination of the actual types in exactly the same way.
\begin{align*}
	\phi_{TSG, Branch} = &\forall e \Big((\Bag_{LR}(e, X) \wedge \Bag_L(e, Y)
	\wedge \Bag_R(e, Y')) \\
	&\to \bigwedge_{\stackrel{i = 1,\ldots,r}{j = 1,\ldots,r}}((e \in C_{i,
	L} \wedge e \in C_{j, R}) \to e \in C_{f_J(\{i, j\}, X), LR})\Big)
\end{align*}
Knowing that all graphs that have property $P$ are contained in one of the
equivalence classes $C_{A_1},\ldots,C_{A_p}$ and that the root bag is always of
type $R2$, we can define $\phi_{Root}$ directly.
\begin{align*}
	\phi_{Root} = \forall X \forall e \Big((\Root(X) \wedge \Bag_{R2}(e, X))
	\to \bigvee_{i = A_1,\ldots,A_p} e \in C_{i, R2} \Big)
\end{align*}

\subsection{Equivalence Class Membership - Generalized}\label{appSecMSOLEQCMG}
In the current section we describe how to define predicates for the equivalence
class membership of (partial) terminal subgraphs in any MSOL-definable ordered
tree decomposition, hence concluding the proof of Lemma \ref{lemFIIDMTD}. In
this case we do not know the specific shape of the tree decomposition, so our case
analysis becomes somewhat more lengthy. We give examples for each predicate
involved from which it will become apparent that one can define any such case
in a similar way.
\\
Once we defined all predicates for MSOL-definable ordered tree decompositions,
we additionally show how to define the case of branch nodes in an MSOL-definable
tree decomposition of bounded degree, hence concluding the proof of Lemma
\ref{lemFIIDMTDBC}.
\par
As before (Appendix \ref{appSecMSOLEQCMH}) we first define all sets that we need
for the predicates and then distinguish the cases that $X$ is an intermediate node
or a branch node. These predicates will be defined in detail in the following
sections.
\begin{align*}
	\phi_{TSG} = &\Big(\exists C_{i, \tau}^V \exists C_{i, \sigma}^E \exists
	C_{i, \tau}^{V|P} \exists C_{i, \sigma}^{E|P} \exists C_{i, \tau}^{V|C} \exists
	C_{i, \sigma}^{E|C}\Big)_{\stackrel[\sigma, \in
	\{\sigma_1,\ldots,\sigma_s\}]{i = 1,\ldots,r}{\tau, \in
	\{\tau_1,\ldots,\tau_t\}}} \\
	&\Big(\phi_{TSG, Int} \wedge \phi_{TSG, Branch}\Big)
\end{align*}

\subsubsection{Intermediate Nodes}\label{appSecMSOLEQCMGI}
First, we define the equivalence class membership for terminal
subgraphs corresponding to an intermediate node in the tree decomposition. We
conduct a case analysis as discussed in the proof of Lemma \ref{lemFIIDMTD}
w.r.t.\ the types of the bags $X$ and $Y$.
\begin{align}
	\phi_{TSG, Int} = &\forall X\forall Y \Big((\Intermediate(X) \wedge
	\Parent(X, Y)) \nonumber \\
	&\to \bigwedge_{\stackrel{\tau, \tau' \in \{\tau_1,\ldots,\tau_t\}}{\sigma,
	\sigma' \in \{\sigma_1,\ldots,\sigma_s\}}} \Big(\phi_{Int,\tau,
	\tau'} \wedge
	\phi_{Int, \sigma, \sigma'} \wedge \phi_{Int, \tau, \sigma}
	\wedge \phi_{Int, \sigma, \tau}\Big)\Big)\label{eqMSOLPhiInt}
\end{align}
Case \ref{lemFIIDMTDCaseV} Both bags belong to a vertex. For each pair of types
$\tau, \tau' \in \{\tau_1,\ldots,\tau_t\}$ one can define the following predicate.
\begin{align*}
	\phi_{Int, \tau, \tau'} = &\forall v \forall v'
	\Big((\Bag_{\tau}^V(v, X) \wedge \Bag_{\tau'}^V(v', Y)) \\	
	&\to \bigwedge_{i = 1,\ldots,r} (v' \in C_{i, \tau'}^V \to v \in
	C_{f_I(i, X), \tau}^V) \Big)
\end{align*}
Case \ref{lemFIIDMTDCaseE} Both bags belong to an edge. For each pair of types
$\sigma, \sigma' \in \{\sigma_1,\ldots,\sigma_s\}$ we can write down a similar predicate.
\begin{align*}
	\phi_{Int, \sigma, \sigma'} = &\forall e \forall e'
	\Big((\Bag_{\sigma}^E(e, X) \wedge \Bag_{\sigma'}^E(e', Y)) \\	
	&\to \bigwedge_{i = 1,\ldots,r} (e' \in C_{i, \sigma'}^E \to e \in
	C_{f_I(i, X), \sigma}^E) \Big)
\end{align*}
Case \ref{lemFIIDMTDCaseVE} The bag $X$ belongs to a vertex and $Y$ belongs to
an edge. For each pair of a type $\tau \in \{\tau_1,\ldots,\tau_t\}$ and $\sigma
\in \{\sigma_1,\ldots,\sigma_s\}$ one can define:
\begin{align*}
	\phi_{Int, \tau, \sigma} = &\forall v \forall e
	\Big((\Bag_{\tau}^V(v, X) \wedge \Bag_{\sigma}^E(e, Y)) \\
	&\to \bigwedge_{i =
	1,\ldots,r} (e \in C_{i, \sigma}^E \to v \in C_{f_I(i, X), \tau}^V)
	\Big)
\end{align*}
Case \ref{lemFIIDMTDCaseEV} The bag $X$ belongs to an edge and $Y$ belongs to a
vertex. For $\sigma, \tau$ as above we define:
\begin{align*}
	\phi_{Int, \sigma, \tau} = &\forall e \forall v
	\Big((\Bag_{\sigma}^E(e, X) \wedge \Bag_{\tau}^V(v, Y)) \\
	&\to \bigwedge_{i =
	1,\ldots,r} (v \in C_{i, \tau}^V \to e \in C_{f_I(i, X), \sigma}^E)
	\Big)
\end{align*}

\subsubsection{Branch Nodes}\label{appSecMSOLEQCMGJ}
In the following we will define predicates for branch nodes, such that all bags
considered always correspond to vertices in the graph. Note that in the cases
that some of them are edge bags, one can write down all predicates in
the same way (replacing some vertices/vertex sets with edges/edge sets in the
predicates).
\\
First we define the general case, in which $Y$ is neither the leftmost nor the
rightmost child of $X$ and deal with the special cases later. Let $Y'$ is
the direct right sibling of $Y$.
\begin{align*}
	\phi_{Branch, \tau, \tau', \tau''}^I = &\forall v \forall v' \forall v''
	\Big((\Bag_{\tau}^V(v, X) \wedge \Bag_{\tau'}^V(v', Y) \wedge
	\Bag_{\tau''}^V(v'', Y')) \\
	\to &\bigwedge_{i = 1,\ldots,r} \Big(\Big(v \in
	C_{i, \tau}^{V|P} \wedge v' \in C_{i, \tau'}^{V|C} \wedge v' \in C_{j, \tau'}^V\Big)
	\\
	&\to \Big(v \in C_{f_J(i, j), \tau}^{V|P} \wedge v'' \in C_{f_J(i,
	j), \tau''}^{V|C} \Big) \Big)\Big)
\end{align*}
Now we consider the situation when $Y$ is the leftmost child of $X$ with
right sibling $Y'$. In this case we derive the partial terminal subgraph
$\pTermSG{X}{Y'}$ by pretending that $Y$ is the only child of $X$ and using
the method for intermediate nodes. It is easy to see that this way we indeed
define the equivalence class membership for $\pTermSG{X}{Y'}$.
\begin{align*}
	\phi_{Branch, \tau, \tau', \tau''}^{L+} = & \forall v \forall v' \forall v''
	\Big((\Bag_{\tau}^V(v, X) \wedge \Bag_{\tau'}^V(v', Y) \wedge
	\Bag_{\tau''}(v'', Y')) \\
	\to &\bigwedge_{i = 1,\ldots,r} \Big(v' \in C_{i, \tau'}^V \to
	\Big(v \in C_{f_I(i, X), \tau}^{V|P} \wedge v'' \in C_{f_I(i, X),
	\tau''}^{V|C}\Big)\Big)\Big)
\end{align*}
When reaching the rightmost child of a branch bag $X$, we derive the terminal
subgraph $\termSG{X}$. Assume in the following that $Y$ is the rightmost child
of $X$.
\begin{align*}
	\phi_{Branch, \tau, \tau'}^{R+} = &\forall v \forall v'
	\Big((\Bag_{\tau}^V(v, X) \wedge \Bag_{\tau'}^V(v', Y)) \\
	\to &\bigwedge_{i = 1,\ldots,r} \Big(\Big(v \in C_{i,
	\tau}^{V|P} \wedge v' \in C_{i, \tau'}^{V|C} \wedge v' \in C_{j,
	\tau'}^V\Big) \to v \in C_{f_J(i, j), \tau}^V\Big)\Big)
\end{align*}
One can define a predicate $\phi_{TSG, Branch}$ in a similar way as
$\phi_{TSG, Int}$ using the predicates described above together with
$\Child_{L+}(X)$, $\Child_{R+}(X)$ and $\oriNBA(X, Y)$. Disregarding the types
of bags for now, one can define the predicate $\phi_{TSG, Branch}'$ in the
following way.
\begin{align*}
	\phi_{TSG, Branch}' = \forall X \forall Y &\Big((\Parent(X, Y) \wedge
	\JoinB(X)) \to \Big(\Big(\Child_{R+}(Y) \wedge \phi_{JoinB}^{R+}\Big) \\
	&\vee \forall Y' \Big(\oriNBA(Y, Y') \to \Big(\Big(\Child_{L+} \wedge
	\phi_{\JoinB}^{L+} \Big) \\
	&\vee \Big(\neg \Child_{L+}(Y) \wedge \phi_{\JoinB}^I\Big)
	\Big)	\Big)	\Big)	\Big)
\end{align*}
Note that to include the case analysis, one can define a predicate
$\phi_{TSG, Branch}$ as it is done in the definition of $\phi_{TSG,
Int}$ (Predicate \ref{eqMSOLPhiInt}), for all combinations of
vertex/edge types.
\subsubsection{Branch Nodes for Bounded Degree Tree
Decompositions}\label{appSecMSOLEQCTDBD}
To finish the proof of Lemma \ref{lemFIIDMTDBC}, we only have to show how to
define a predicate for branch nodes with a constant number of children as
explained in the proof. \\
Again, we give an example predicate for the case that all bags involved are
vertex bags and note that all other cases can be defined similarly. Consider a
branch bag $X$ with child bags $X_1,\ldots,X_k$, all corresponding to vertices
in the graph and types $\tau_1,\ldots,\tau_k$. Then we can define this predicate as
follows.
\begin{align*}
	\phi_{Branch, \tau, \tau_1,\ldots,\tau_k} = &\forall v \forall v_1 \cdots
	\forall v_k \Big((\Bag_\tau^V(v, X) \wedge \Bag_{\tau_1}^V(v_1, X_1) \\ 
	&\wedge \cdots \wedge \Bag_{\tau_k}^V(v_k, X_k)) \to
	\bigwedge_{\stackrel[i_k = 1,\ldots,r]{i_1 =
	1,\ldots,r}{\cdots}}\Big((v_1 \in C_{i_1, \tau_1}^V \\
	&\wedge \cdots \wedge v_k \in C_{i_k, \tau_k}^V) \to v \in
	C_{f_J(\{i_1,\ldots,i_k\}, X), \tau}^V\Big)\Big)
\end{align*}

\subsection{$k$-Cycle Trees}\label{appSecMSOLLG}
In the current section we give all predicates to define a tree decomposition of
a $k$-cycle tree in MSOL, as explained in the proof of Lemma \ref{lemkCycTrees}.
We first define the edge orientation $\msolOri$ and then all predicates for the
bag types. Note that since this construction is very
similar to the one for Halin graphs, we do not define the
$\Parent$-predicate explicitly, as it works in almost the exact same way. \par
As a first step we define a predicate to check whether two vertices have a
certain (constant) distance in a given edge set.
\begin{align*}
	\dist(v, w, E') = k \Leftrightarrow (\exists E_P \subseteq E') (\Path(v, w,
	E_P) \wedge |E_P| = k)
\end{align*}
This allows us to define the the $i$-th cycle of the graph.
\begin{align*}
	E' = \Cycle_i \Leftrightarrow & \dir{\Cycle}(\IncV(E'), E') \wedge \forall v
	(\Inc(v, E') \to \dist(c, v, E) = i)
\end{align*}
We can write down the orientation $\msolOri$ described in the proof of Lemma
\ref{lemkCycTrees} in the following way.
\begin{align*}
	\msolOri = &\exists E_T \exists E_{C_1}\cdots \exists
		E_{C_k}\exists r_1\cdots \exists r_{k-1}\Big((\Part_E(E, E_T,
		E_{C_1},\ldots,E_{C_k}) \\ 
		&\wedge \dir{\Tree}(V, E_T) \wedge \bigwedge_{i = 1,\ldots,k}E_{C_i} =
		\Cycle_i \\
		&\wedge \bigwedge_{i = 1,\ldots,k-1} \Big(r_i \in \IncV(E_{C_i}) \wedge
		\dist(r, r_i, E_T) = k-i \Big)\Big)
\end{align*}
We can define a predicate $\oriNB^i(e, f)$ in complete analogy to $\oriNB(e, f)$
as shown in Appendix \ref{appSecMSOLHalinChOrd} by simply replacing $E_C$ by
$E_{C_i}$ and $r$ by $r_i$ (for the case that $i=k$ don't have to modify it).
This predicate is true if and only if $e$ is on the left of $f$, such that $e$
and $f$ have the same head vertex, i.e. their tail vertices lie on the same
cycle.
\par
Now we turn to defining the $i$-th boundary vertex (Definition \ref{defIBound}).
\begin{align}
	w = bd_i^r(v, E') \Leftrightarrow &(w \in V_{C_i} \wedge v = w) \nonumber\\
	&\vee (\exists E_P \subseteq E') (\dir{\Path}(v, w, E_P) \wedge w \in
	\IncV(E_{C_i})) \nonumber \\
	&\wedge (\forall e \in E_P) \neg(\exists E_P' \subseteq E') \Big(\dir\Path(v,
	w, E_P') \wedge (\exists e' \in E_P') \nonumber \\
	&\bigvee_{i = 1,\ldots,k}\oriNB^i(e, e')\Big) \label{predIBound}
\end{align}
To define $bd_i^l$, we simply replace $\oriNB^i(e, e')$ by $\oriNB^i(e', e)$ in
line \ref{predIBound}. In the following we abbreviate $w = bd^i(v, E_T)$
to $w = bd^i(v)$.
We denote by $NB_R(e)$ the edge set containing $e$ and all its right neighbor
edges. \par
We are now equipped with all tools to define the bag types for a tree
decomposition of a $k$-cycle tree. We use the same notation as in Appendix
\ref{appSecMSOLHalinTD}, that is, we have an edge $e = \{x, y\}$, such that $y$
is the parent of $x$ in $E_T$ and assume that the vertex $y$ lies on cycle
$C_i$.
The predicate $\CarryBdR$ defines the case that the vertex $x$ does not have a
left boundary on a cycle $C_j$, so that we have to pass on the right boundary
vertex of $y$ without the edge $e$ and its right neighbors.
\begin{align*}
	\CarryBdR(e, z)_j \Leftrightarrow (\neg(\exists z' (z' = bd_j^l(x))) \wedge z =
	bd_j^r(y, E_T \setminus NB_R(e))
\end{align*}
We continue by defining the bag types $R1,\ldots,LR$.
\begin{align*}
	\Bag_{R1}(e, X) \Leftrightarrow &z \in X \leftrightarrow 
	\bigvee_{i < j \le k}\Big(z = bd_j^l(x) \vee z = bd_j^r(x)\Big) \\
	\Bag_{R2}(e, X) \Leftrightarrow &z \in X \leftrightarrow \Big(z = y \vee
	\bigvee_{i < j \le k}\Big(z = bd_j^l(x) \vee z = bd_j^r(x)\Big)\Big)
\end{align*}
\begin{align*}
	\Bag_{L1}(e, X) \Leftrightarrow &z \in X \leftrightarrow \Big(z = y
	\vee \bigvee_{i \le j \le k}\Big(z = bd_j^l(y) \\
	&\vee z = bd_j^r(y, E_T
	\setminus NB_R(e)) \Big)\Big) \\
	\Bag_{L2}(e, X) \Leftrightarrow &z \in X \leftrightarrow \Big(z = y 
	\vee \bigvee_{i < j \le k}\Big(z = bd_j^l(y) \vee z = bd_j^l(x) \\
	&\vee z = bd_j^r(y, E_T \setminus NB_R(e)) \Big)\Big)
\end{align*}
\begin{align*}
	\Bag_{L3}(e, X) \Leftrightarrow &z \in X \leftrightarrow \Big(z = y
	\vee \bigvee_{i < j \le k}\Big(z = bd_j^l(y) \vee z = bd_j^l(x)\\
	&\vee \CarryBdR(e, z) \Big)\Big) \\
	\Bag_{LR}(e, X) \Leftrightarrow &z \in X \leftrightarrow \Big(z = y
	\vee \bigvee_{i < j \le k}\Big(z = bd_j^l(x) \vee z = bd_j^r(x) \\
	&\vee z = bd_j^l(y) \vee \CarryBdR(e, z) \Big)\Big)
\end{align*}
Note that defining the $\Parent$-predicate works in the same way
as for Halin graphs, taking into account the missing bag type $R3$.

\subsection{Adding Feedback Edge/Vertex Sets}\label{appSecMSOLFESFVS}
In this section we complete the proof of Theorem \ref{thmFESFVS}. 
In the following, let $G' = (V', E')$ and $G = (V, E)$ be graphs as stated in
Theorem \ref{thmFESFVS}. Assume that we are given predicates $\Bag_{\tau}^V$
and $\Bag_\sigma^E$ for vertex bag types $\tau$ and $\sigma$, defined for vertices
and edges of the spanning tree $E_T$ of a graph, defining a tree decomposition
of $G'$.
 One can observe that we can define the sets $V'$ and
(a set representing) $E'$ easily, using the following facts.
\begin{itemize}
  \item Each vertex $v' \in V'$ is contained in a bag of the tree decomposition,
  	i.e.\ (at least) one of the $\Bag$-predicates evaluates to \emph{true} for
  	some set $X \subseteq V$.
  \item For each edge $e' \in E'$ there is a bag containing both endpoints. Note
  	that if there is an edge in $E \setminus E'$, such that both its endpoints
  	are contained in a bag, we do not need to consider it any further. 
\end{itemize}
In the following we assume that $V'$ and $E'$ are defined and $\FundCyc$ uses
the maximal spanning tree $E_T$, upon which the construction of the tree
decomposition of $G'$ is based.
First, we consider the case of feedback edge sets. We use the notion of
fundamental cycles rather that directly referring to biconnected components,
since it makes our predicate shorter (while in this case they express the same
thing).\footnote{Note that the predicate $\FundCyc$ can easily be defined for a
combination of a vertex and an edge as well.}
\begin{align*}
	\Bag^{E,+}_{\sigma}(e, X) \Leftrightarrow &v' \in X \leftrightarrow
	\Big((\exists e' \in E \setminus E')\Big(\Inc(v', e') \wedge \FundCyc(e, e') \\
	&\wedge \forall w ((\neg v' = w \wedge \Inc(w , e')) \to \col(v') <
	\col(w))\Big)\Big) \\
	\Bag^{V, +}_\tau(v, X) \Leftrightarrow &v' \in X \leftrightarrow \Big((\exists
	e' \in E \setminus E')\Big(\Inc(v', e') \wedge \FundCyc(v, e') \\
	&\wedge \forall w ((\neg v' = w \wedge \Inc(w , e')) \to \col(v') <
	\col(w))\Big)\Big)
\end{align*}
For feedback vertex sets we can define similar additions to the respective
predicates, directly using the biconnected components mentioned in the proof.
\begin{align*}
		\Bag_\tau^{V, +}(v, X) \Leftrightarrow &v' \in X \leftrightarrow \Big(
		(\exists V_2 \subseteq V)(v \in V_2 \wedge v' \in V_2 \\
		&\wedge \Conn_2(V_2, E_T \cup \IncE(V_2 \setminus V')))\Big) \\
		\Bag_\sigma^{E, +}(e, X) \Leftrightarrow &v' \in X \leftrightarrow \Big(
		(\exists V_2 \subseteq V)\Big(v' \in V_2 \wedge e' \in \IncE(V_2 \setminus V')
		\\
		&\wedge \Conn_2(V_2, E_T \cup \IncE(V_2 \setminus V'))\Big) \Big)
\end{align*}

\subsection{Bounded Vertex and Edge Remember Number}\label{appSecMSOLVRER}
As the last of our extensions, we show how to define tree decompositions
that have a bounded vertex and edge remember number. Hence, we
will conclude the proof of Theorem \ref{thmVRER}, which we used to prove the
case for bounded degree $k$-outerplanar graphs. \par
First, we are going to show how to identify an edge set as a spanning tree with
vertex remember number less than or equal to $\kappa$ and edge remember number
less than or equal to $\lambda$, both constant.
\begin{align*}
	\exists E_T &(\Tree(V, E_T) \wedge vr(E_T) \le \kappa \wedge er(E_T) \le
	\lambda) \\
	vr(E_T) \le \kappa \Leftrightarrow &(\forall v \in V)(\forall e_1 \in E
	\setminus E_T)\cdots \forall (e_{\kappa + 1} \in E \setminus E_T)\\
	&\Big(\Big(\bigwedge_{i = 1,\ldots,\kappa + 1} \FundCyc(v, e_i)\Big) \to
	\bigvee_{1 \le i < j \le \kappa + 1} e_i = e_j\Big) \\
	er(E_T) \le \lambda \Leftrightarrow &(\forall e \in E)(\forall e_1 \in E
	\setminus E_T)\cdots \forall (e_{\lambda + 1} \in E \setminus E_T)\\
	&\Big(\Big(\bigwedge_{i = 1,\ldots,\lambda + 1} \FundCyc(e, e_i)\Big) \to
	\bigvee_{1 \le i < j \le \lambda + 1} e_i = e_j\Big)
\end{align*}
In the following, assume that $E_T$ is the edge set of the spanning tree of $G$
(as shown above), which additionally has edge orientations, defined in MSOL by
predicates $\head$ and $\tail$ (cf. Appendix \ref{appSecMSOLHalinOrd}). Note
that the last predicate in the list, $\oriNB(X_a, X_b)$ requires an ordering on
edges with the same head vertex.
\begin{align*}
	\Bag_V(v, X) \Leftrightarrow &v' \in X \leftrightarrow (v' = v \vee (\exists
	e \in E \setminus E_T)(\Inc(v', e) \\ 
	&\wedge \FundCyc(v, e))) \\
	\Bag_E(e, X) \Leftrightarrow &v' \in X \leftrightarrow (\Inc(v', e) \vee
	(\exists e' \in E \setminus E_T)(\Inc(v', e') \\ 
	&\wedge \FundCyc(e, e'))) \\
	\Parent(X_p, X_c) \Leftrightarrow &\exists v(\exists e \in E_T)((\Bag_V(v,
	X_p) \wedge \Bag_E(e, X_c) \wedge \head(v, e)) \\ 
	&\vee (\Bag_V(v, X_c) \wedge \Bag_E(e, X_p) \wedge \tail(v, e))) \\
	\oriNB(X_a, X_b) \Leftrightarrow &(\exists e_a \in E_T)(\exists e_b \in
	E_T)(\head(e_a) = \head(e_b) \wedge \oriNB(e_a, e_b))
\end{align*}

\end{document}